\documentclass[10pt]{IEEEtran}

\usepackage{fullpage,amsthm,amsmath,citesort,graphicx}
\usepackage{amssymb,epsf,latexsym}
\usepackage{amstext,bm}
\usepackage{graphicx}
\usepackage{subfigure}

\usepackage{setspace}

\newtheorem{theorem}{Theorem}
\newtheorem{lemma}[theorem]{Lemma}
\newtheorem{proposition}[theorem]{Proposition}
\newtheorem{prop}[theorem]{Proposition}
\newtheorem{corollary}[theorem]{Corollary}
\newtheorem{conjecture}[theorem]{Conjecture}

\newtheorem{remark}[theorem]{Remark}

\newtheorem{definition}{Definition}
\newenvironment{proofof}{\indent \indent {\em Proof of}}{\hfill$\bf  \blacksquare$\bigskip}

\newcommand{\LL}{\mathit{L}}

\newcommand{\hD}{\widehat D}
\newcommand{\hhD}{\widehat{\widehat D}}

\newcommand{\mH}{\mathcal{H}}
\newcommand{\mf}{\mathcal{F}}
\newcommand{\mE}{\mathcal{E}}

\newcommand{\be}{\begin{eqnarray}}
\newcommand{\ben}{\begin{eqnarray*}}
\newcommand{\ee}{\end{eqnarray}}
\newcommand{\een}{\end{eqnarray*}}

\IEEEoverridecommandlockouts
\makeatletter

\newcommand{\Rmnum}[1]{\expandafter\@slowromancap\romannumeral #1@}
\makeatother

\newcommand{\post}[2]{
\centering \leavevmode
\includegraphics[width=#2cm]{#1}
 }

\begin{document}
\title{Stability of a Peer-to-Peer Communication System }

 \author{
\authorblockN{Ji Zhu and Bruce Hajek}\\
\authorblockA{Department of Electrical and Computer Engineering \\
 and the Coordinated Science Laboratory  \\
 University of Illinois at Urbana-Champaign}
 \thanks{This research was supported in part by the National Science Foundation under grant NSF Grant CCF 10-16959.
An earlier version of this work was presented at the {\em 2011 ACM Symposium on Principles of Distributed Computing.}
}  
\thanks{Copyright (c) 2012 IEEE. Personal use of this material is permitted.ÊÊHowever, permission to use this material for any other purposes must be obtained from the IEEE by sending a request to pubs-permissions@ieee.org.}
}  
\maketitle

\begin{abstract}
This paper focuses on the stationary
portion of file download in an unstructured peer-to-peer
network, which typically follows for many hours
after a flash crowd initiation. The model includes
the case that peers can have some pieces at the time
of arrival.   The contribution of
the paper is to identify how much help is needed from
the seeds, either fixed seeds or peer seeds (which
are peers remaining in the system after obtaining a
complete collection) to stabilize the system.
The dominant cause for instability is the
missing piece syndrome, whereby one piece becomes
very rare in the network.  It is shown that stability can
be achieved with only a small amount of help from
peer seeds--even with very little help from a fixed
seed, peers need dwell as peer seeds on average only
long enough to upload one additional piece.
The region of stability is insensitive to the
piece selection policy.  Network coding can
substantially increase the region of stability in
case a portion of the new peers arrive with
randomly coded pieces.

\vspace{.1in}
\end{abstract}

\section{Introduction}
Second generation P2P networks such as
\textit{BitTorrent} \cite{Cohen03}, divide a file to be distributed into distinct pieces and enable peers (or clients) to share these pieces efficiently. BitTorrent, with its rarest first and choke algorithms \cite{Cohen03,legout2006rarest},  has been shown in practice to scale well with the number of participating peers \cite{YangDeVeciana04,QiuSrikant04,Menasche_etal10,legout2006rarest,Tewari,GkantsidisRodriguez05,karagiannis2004p2p}.


Understanding how a BitTorrent like P2P system works over a long period of time is difficult, due to the following details.
Each peer maintains a set of neighbors it can connect with. According to the choking algorithm,  a peer unchokes three
neighbors from which the peer has the fastest download rate, at the same time it also unchokes a randomly chosen neighbor which has pieces needed by the peer. The choking algorithm works as a distributed peer selection mechanism to continuously shape the topology of the network; it is influenced by heterogeneous link speeds and by the sets of pieces available at different peers.
Peers track the pieces available at their neighbors and the selection of pieces to be downloaded is biased towards
the rarest pieces first.
Consequently, analytical models capturing all aspects of BitTorrent in detail are intractable. Simulations have revealed
extensive insight about the scalability, robustness, and efficiency of P2P networks, but simulations alone can cover only
a small portion of the range of parameter values and network settings.   Analysis complements simulations by helping to
identify potential pitfalls and as a means to understand and avoid them.

The following stochastic model of P2P networks is examined in \cite{HajekZhu10_isit,HajekZhu10_full}.
A seed uploads at a constant rate $U_s$; peers arrive as a rate $\lambda$ Poisson process; the seed and peers apply
uniform random peer selection and diverse piece selection policies; each peer leaves as soon as it has all pieces.
It is shown in \cite{HajekZhu10_isit,HajekZhu10_full} that the stability region is governed
by the {\em missing piece syndrome}. The missing piece syndrome is an abnormal condition appearing when  there are many peers in the system and all of them are missing the same piece. Such a large group of peers missing the same piece
severely limits the spread of the piece in the network.  Peers without the missing piece quickly join the group and peers
with the missing piece quickly depart. The main result in \cite{HajekZhu10_isit,HajekZhu10_full} is that the network
may never recover from the missing piece syndrome if the upload rate of the seed is less than the arrival rate of new peers,
and the network is positive recurrent if the upload rate of the seed is smaller than the arrival rate of new peers.



This paper extends the basic results of \cite{HajekZhu10_isit,HajekZhu10_full} in two particular ways:  peers can already
have some pieces at the time of their arrival, and peers can dwell awhile in the network after obtaining a complete collection.
The main result in this paper, Theorem \ref{thm:main},
 provides the stability region of the network within the space of values of arrival rates, seed uploading capacity, and peer dwelling time. The proof of the main result is shaped by showing that the system either is trapped by
the missing piece syndrome, or that it always escapes the missing piece syndrome,  depending on the parameter values.
This paper reveals the least amount of time peers must dwell after obtaining the entire file so that the whole network is positive recurrent.   A corollary of our result is that if each peer can upload one additional piece after obtaining the whole file before
departing, the network is stable under any positive seed uploading capacity and any arrival rates. In BitTorrent, the size of a
single piece is typically a small fraction of the entire file (about 0.5\%) so that it is a light burden for a peer to dwell
in the network long enough to upload one more piece after obtaining a complete collection.
The proof techniques are similar to those used in \cite{HajekZhu10_isit,HajekZhu10_full}, but are
modified to handle the more general model here.
For the proof of the positive recurrence for other parameter values, a Lyapunov function is used as in  \cite{HajekZhu10_isit,HajekZhu10_full}, but it is no longer quadratic, and a variation of the standard
big "O" notation is introduced.  There are quadratic terms in the
Lyapunov function, but some related terms
are added to cover the case that sufficient downloading capacity has to build up as new arrivals bring
new pieces with them.  

Four extensions to Theorem 1 are also presented
in this paper.  The first extension is to  point out that Theorem 1 remains true for a wide variety of piece
selection policies, as long
as they select useful pieces when present, and the same uniform, random peer selection policy is used.
The second extension is to point out how Theorem 1 can be modified to incorporate network coding.
Such an extension was also given in \cite{HajekZhu10_full} for the less general model there, which
specified, in particular, that peers have no pieces when they arrive.  In that context it was shown in
\cite{HajekZhu10_full} that network coding does not increase the region of stability of the peer to peer
system.   In contrast, we find here that when peers arrive with some
(randomly coded) pieces, network coding substantially increases the region of stability.  The third extension
addresses variations of the model such that the time between two consecutive transfer attempts is reduced
if there is no useful piece to transfer.  The fourth
extension is to consider the borderline case, between the necessary and sufficient conditions of
Theorem 1. 

The organization of the paper is as follows.   Related work is presented in Section \ref{sec:related}.
The network model and Theorem 1, the main result of this paper, are described in Section \ref{sec:Model}.
Section \ref{sec:example} presents three examples that illustrate Theorem 1.  Section \ref{sec:outline_of_proof}
presents an outline of the proof of Theorem 1, while the detailed proof itself is given in Sections
\ref{sec:transience} and \ref{sec:positiverecurrence}, which prove the transience and positive recurrence
parts of Theorem 1, respectively.   The extensions to Theorem \ref{thm:main} are given in
Section \ref{sec:extensions}, and a brief conclusion is given in Section \ref{sec:conclusion}.
Miscellaneous results used in the main part of the paper are summarized in the appendix.

\section{Related Work}  \label{sec:related}

This section briefly points to work related to stability and the missing piece syndrome in BitTorrent like P2P networks with
models similar to the one here.   Like this paper, the paper of Massouli\'{e} and Vojnovic \cite{MassoulieVojnovic08} assumes that peers having various collections of pieces arrive according to Poisson processes, although there is no seed.   The analysis given in  \cite{MassoulieVojnovic08} is based on scaling the initial state and the arrival rates by a parameter that goes to infinity.
The asymptotic analysis gives rise to a fluid limit, described by a vector ordinary differential equation.
The existence of a symmetric equilibrium point of the fluid limit is established.
Like this paper, the paper of Leskel\"{a}, Robert, and Simatos  \cite{LeskelaRobertSimatos10} considers the case
of each peer dwelling awhile after it has obtained a complete collection.   The case in which a file is not divided at all,
and the case in which a file is divided into two pieces that must be collected by all peers in the same order, are considered, and
the required mean dwell time is identified for stabilizing the system.
Models in \cite{MassoulieVojnovic08,HajekZhu10_isit,HajekZhu10_full,LeskelaRobertSimatos10} are discussed as special cases of the model in this paper.

Two-piece P2P models under slightly different
assumptions are studied in  \cite{NorrosReittuEirola09}, and essentially the same stability condition as in \cite{HajekZhu10_isit,HajekZhu10_full} is obtained for the two-piece special case.
By modeling BitTorrent as multiple $M/G/\infty$ queues, the authors in \cite{Menasche_etal10} provide closed form steady state distributions and study the self-sustainability of their systems. In the simulation of \cite{Menasche_etal10}, the authors find their ``smooth download assumption" and ``swarm sustainability" break down if the seed upload capacity is small; this is evidence of the missing piece syndrome.

The BitTorrent choking algorithm has attracted considerable interest from researchers, due to its ability to encourage reciprocity and increase scalability.   Based on experiments for the case of flash crowds in BitTorrent, the authors in \cite{legout2006rarest}
concluded that the choke algorithm and rarest first piece selection together can foster reciprocation and guarantee close to ideal diversity of the pieces among peers. It is worth noting that the experiment in \cite{legout2006rarest} about transient states,  which appear because of the upload constraint of the seed,  gives evidence of the missing piece syndrome. In \cite{legout2007clustering}, the authors show  that the choking algorithm can facilitate the formation of clusters of similar-bandwidth peers. The authors measured the performance of BitTorrent protocols on a PlanetLab platform, and discovered that when the seed upload capacity is high, peers mainly upload
to other peers with roughly the same bandwidth.   But when the seed upload capacity is low, such clustering of peers
does not emerge.
In \cite{MenascheInfocom}, the authors compare direct reciprocity, where users exchange contents directly, and indirect reciprocity, where users upload contents based on credits of their targets. They show that an indirect reciprocity schedule can be replaced by a direct reciprocity schedule with a loss of efficiency at most a half if users can restore undemanded contents for bartering. They also provide simulations showing the benefits of having a public board which announces the content distribution and having a matchmaker which pairs users together by a maximum weight matching algorithm. 

Papers \cite{5061908,zhou2011stability,ioannidis2010distributed} concern concurrent delivery of multiple files in a P2P network. Peers can store files they do not request in order to increase reciprocation and efficiency of file distribution. Models about single-piece file sharing through mobile networks are studied in \cite{ioannidis2010distributed,zhou2011stability}.   In \cite{zhou2011stability} the authors suppose multiple single-piece files are to be downloaded by some of the peers, and peers store and exchange files they do not request. Assuming Poisson arrivals and random peer contact, the authors establish fluid limits for a broad family of file exchanging policies, and derive the stability region for a static-case policy (peers do not exchange files unless they can get their requested files). They further show that by mixing multiple swarms together the network scalability is increased in the sense that only one swarm can become unstable.  In \cite{5061908} the authors discuss multiple-channel live streaming and show how the performance increases if some peers can apply their spare capacity to distribute channels they are not watching. Papers \cite{tomozei2010flow,zhou2007simple} are also about live streaming by P2P networks. In \cite{zhou2007simple} the authors provide a simple queue model to compare rarest first and greedy piece selection policies in P2P live streaming, and propose a mixed selection policy to balance the trade-off between start-up latency and continuity.

Network coding can improve the network performance. Network coding was first proposed in \cite{AhlswedeCaiLiYeung},
where it is shown that a sender can communicate information to a set of receivers if  the min-cut max-flow bound is satisfied for connections to each receiver. Simulations with network coding applied in P2P file distribution described in \cite{GkantsidisRodriguez05}  show that in a P2P network under topologies with bad cuts, network coding can provide a much higher average file distribution rate than that provided without coding or with source coding only.  Better robustness also appears when network coding is simulated on a P2P network with dynamic arrivals and departures. In \cite{DebMedardChoute06} the authors study a gossip model under random linear coding, with each peer initially having a single unique piece, and all peers are to 
collect all pieces, and peers are assumed to apply random contact and transmit random linear combinations of the messages they
own to their targets. It is shown in \cite{DebMedardChoute06} that with network coding, the gossip can be completed in time proportional to the number of peers, with high probability.
The paper \cite{tomozei2010flow} focuses on the efficiency of network coding for P2P live streaming.   It shows
that when network coding is applied and a distributed, stochastic version of a primal-dual algorithm is used, then a
fluid scale limit admits a cost optimal operating point as a fixed point.
Network coding is considered in \cite{HajekZhu10_full} for the assumptions of that paper (peers arrive with no
pieces, there is a fixed seed, and peers depart after obtaining a complete collection).  In that context, while network
coding eliminates the need for peers to exchange lists of pieces, the condition for stability is nearly the
same as for random useful piece selection without network coding.

\section{Model and Result \label{sec:Model}} 
The model discussed in this paper is a combination of related models in \cite{MassoulieVojnovic08,QiuSrikant04,YangDeVeciana04}. 
In a single fixed seed P2P network, a large file is divided into $K$ pieces, for some $K\geq 1,$
which are stored in the fixed seed. The fixed seed is not considered to be a peer.
Each  peer in the system holds some subset of the pieces. For any subset $C$ of the total collection of pieces $\{1,2,...K\}$, a peer holding the collection of pieces $C$ is called a {\em type $C$ peer}. In some real P2P networks, peers can get some pieces from a tracker upon their arrival for initialization. To capture that case, we assume type $C$ peers arrive into the system at times of a Poisson process with rate $\lambda_C$. Although we consider all possible values of $(\lambda_C, C\in\mathcal{C})$, typically in practice, $\lambda_C$ is small or equal to zero when $|C|>1$. 

The fixed seed and all peers use the {\em random peer contact} and {\em random useful piece selection} strategies at instants of Poisson processes, with the contact-upload rate of the fixed seed denoted by $U_s$ and the contact-upload rate of any peer denoted by  $\mu, \mu>0$. Specifically, suppose the fixed seed and each peer maintain internal Poisson clocks; the clock of the fixed seed ticks at rate $U_s$, and the clock of any peer ticks as rate $\mu$. Whenever the clock of the fixed seed ticks, the fixed seed contacts a peer, say peer $A$, which is selected uniformly from among all peers. According to the random useful piece selection strategy, the fixed seed checks to see if $A$ needs any pieces, and uploads to $A$ the copy of one piece uniformly chosen from among the pieces needed by $A$. If $A$ does not need any pieces (because $A$ is a peer seed), no piece is uploaded and the  fixed seed remains silent between clock ticks.

A peer similarly uploads pieces. When its rate $\mu$ Poisson clock ticks, it contacts a peer selected at random, and checks to see whether it has pieces needed by the contacted peer. If the answer is yes,  it uploads to the contacted peer a copy of a piece uniformly chosen from among its pieces needed by the contacted peer; if the answer is no,  no piece is uploaded and the peer does not upload pieces between clock ticks. The peer contacts and piece uploads of the fixed seed and peers are assumed to be instantaneous. 

In a real P2P network, peers may upload two or more pieces to different peers at the same time, and peer selection, peer contact and piece upload are not instantaneous. For mathematical simplification, we consider a homogeneous network with the maximum number of upload links of each peer limited to one,   and apply the waiting times of Poisson clocks to model the total time consumed for peer selection, contact, and piece upload. So $1/\mu$ and $1/U_s$ are approximately the average piece transmission time from peer to peer and from the fixed seed to peer in a real P2P network.

Assume that each peer, after becoming a peer seed,
dwells in the system for an exponentially distributed length of time with mean $1/\gamma$, with $0<\gamma\leq \infty$. The case $\gamma=\infty$ is shorthand notation for the case that peers depart immediately after collecting all pieces.
Intuitively, smaller values of $\gamma$ yield better system performance, because peer seeds can upload more pieces if they stay in the system longer. Our result identifies the smallest mean peer seed dwelling time (i.e. largest $\gamma$) sufficient for a stable system. If the rate $U_s$ of the fixed seed is sufficiently large, or if the rates $\lambda_C$ are large enough for some nonempty $C$, the system can be stable even if peers do not become peer seeds (i.e. even if $\gamma=\infty$).
The arrivals of new peers, the peer seed dwell times, and the ticking of Poisson clocks, are mutually independent. The notation and assumptions of the model are summarized as follows:

\begin{itemize}
\item{$\cal C:$} Set of all subsets of $\mathcal{F} = \{1, \ldots , K\}$, where $K\geq 1$ is the number of pieces, and $\mathcal{F}$ is the collection of all pieces.
\item{Type $C$ peer:} A peer with set of pieces $C\in\mathcal{C}$ is a type $C$ peer, which  becomes a type $C \cup \{i\}$ peer if the seed or another peer uploads piece $i\not\in C$ to it. A type $\mathcal{F}$ peer is also called a peer seed.
\item{Type $C$ group:} The set of type $C$ peers in the system.
\item{Arrivals:} Exogenous arrivals of type $C$ peers form a rate $\lambda_C\in[0,\infty)$  Poisson process. To avoid triviality, assume the total arrival rate of peers --- $\lambda_{total} = \sum_{C:C\in\mathcal{C}} \lambda_C$ --- is strictly positive.
Also, without loss of generality, if $\gamma=\infty,$ assume $\lambda_{\mathcal{F}}=0.$
\item{Random peer contact:} The fixed seed contacts a uniformly chosen peer at instants of a Poisson process with rate $U_s\in[0,\infty)$. Every peer contacts a uniformly chosen peer at instants of a Poisson process with rate $\mu\in(0,\infty)$.
\item{Random useful piece upload:} When $A$ contacts $B$, if $B$ does not have all pieces that $A$ has, $A$ uploads to $B$ a copy of one piece uniformly chosen  from among the pieces $A$ has but $B$ does not have. Otherwise no piece is uploaded.
\item{Departures:} If $\gamma \in (0,\infty)$, every peer becomes a peer seed after obtaining all $K$ pieces, and subsequently remains in the system  for an exponentially distributed length of time with mean $1/\gamma$ before departing. If $\gamma=\infty$, then $\lambda_\mathcal{F} = 0$ and peers depart immediately after obtaining all $K$ pieces.
\end{itemize}

Under the assumptions above, the system is a Markov chain with state vector $\mathbf{x} = (x_C:C\in \mathcal{C})\in \mathbb{Z}_+^{|\mathcal{C}|}$ if $\gamma\in(0,\infty)$, and $\mathbf{x}=(x_C:C\in \mathcal{C}-\{\mf\})\in \mathbb{Z}_+^{|\mathcal{C}|-1}$ if $\gamma=\infty$, where $x_C$ is defined to be the number of type $C$ peers, except we define $x_C=0$ in the case $C=\mf$ and $\gamma=\infty$. Define $\Gamma_{C,C'}$ for $C,C'\in\mathcal{C}$ as follows:
\be \label{eq:defGamma}
\Gamma_{C,C'} : = 
\frac{x_C}{n}  \left(\frac{U_s}{K-|C|}+ \mu \sum_{S: i\in S\in\mathcal{C}}
\frac{x_S}{|S-C|}   \right)
\ee
if $n\geq 1$ and $C'=C \cup\{i\}\text{ for some } i\in\mf-C,$  and $\Gamma_{C,C'}=0$ else, 
where $n: = \sum_{C:C\in\mathcal{C}} x_C$ is the total number of peers. In words, unless $C'=\mathcal{F}$ and $\gamma=\infty$, $\Gamma_{C,C'}$ is the aggregate rate of transition of peers from type $C$ to type $C'$; If $C'=\mathcal{F}$ and $\gamma=\infty$, $\Gamma_{C,C'}$ is the aggregate rate of departures from the system of peers of type $C$.

Let $\mathbf{e}_C$ denote the vector with the same dimension as $\mathbf{x}$, with a one in position $C$ and other coordinates equal to zero.
The positive entries of the generator matrix $Q= (q(\mathbf{x},\mathbf{x'}) )$ are given by:
\begin{itemize}
\item if $\gamma\in(0,\infty)$, $\mathbf{x} = (x_C:C\in \mathcal{C})$,
\begin{eqnarray*}
q(\mathbf{x} , \mathbf{x}+\mathbf{e}_C ) & = &  \lambda_C   \\
q(\mathbf{x}, \mathbf{x}-\mathbf{e}_{\mathcal{F}})&=& \gamma x_{\mathcal{F}}\\
q( \mathbf{x},  \mathbf{x} - \mathbf{e}_C+\mathbf{e}_{C\cup \{i\}} ) &= &  \Gamma_{C,C\cup\{i\}},   \mbox{if} ~i\notin C.
\end{eqnarray*}
\item if $\gamma=\infty$, $\mathbf{x}=(x_C:C\in \mathcal{C}-\{\mf\})$,
$$
\begin{array}{rl}
q(\mathbf{x} , \mathbf{x}+\mathbf{e}_C ) & =  \lambda_C   \\
q( \mathbf{x},  \mathbf{x} - \mathbf{e}_C+\mathbf{e}_{C\cup \{i\}} ) &=  \Gamma_{C,C\cup\{i\}},   \mbox{if} \\
&~~~~~|C|\leq K-2,i\notin C.\\
q( \mathbf{x},  \mathbf{x} - \mathbf{e}_C) & =  \Gamma_{C,\mf},  ~\mbox{if} ~ |C|= K-1.
\end{array}
$$
\end{itemize}

The following theorem, which is the main result of this paper, describes the stability region of the P2P system.

\begin{theorem} \label{thm:main}
Let $U_s\in[0,\infty)$, $\mu\in(0,\infty)$, $\gamma\in(0,\infty]$, $\{\lambda_C: C\in\mathcal{C},\lambda_C\in[0,\infty)\}$ with $\lambda_\mf=0$ if $\gamma=\infty$, and $\lambda_{total}>0$ be given. \\
(a) The Markov process with generator matrix $Q$ is transient if either of the following two conditions is true:
\begin{itemize}
\item $0<\mu<\gamma\leq\infty$ and for some $k\in\mathcal{F}$,
\begin{eqnarray}
\lambda_{total}
>\frac{ U_s+\sum_{C: k\in C} \lambda_C(K+1-|C|) }{1-\frac{\mu}{\gamma}} \label{eq:qcd}
\end{eqnarray}
\item $0<\gamma\leq\mu$ and for some piece $k\in\mathcal{F}$, no copies of piece $k$ can enter the system.
\end{itemize}
(b) Conversely, the process is positive recurrent and $E[N]<\infty,$  where $N$ denotes a random variable with
the stationary distribution of number of peers in the system,  if either of the following two conditions is true:
\begin{itemize}
\item $0<\mu<\gamma\leq \infty$ and for any $k\in\mathcal{F}$,
\begin{eqnarray}
\lambda_{total} < \frac{U_s+\sum_{C: k\in C} \lambda_C(K+1-|C|) }{1-\frac{\mu}{\gamma}}.\label{eq:pcd}
\end{eqnarray}
\item $0<\gamma\leq\mu$ and for any $k\in\mathcal{F}$, it is possible for new copies of piece  $k$  to enter the system.
\end{itemize}
\end{theorem}

We remark that when we say new copies of piece $k$ can enter the system, we mean $U_s>0$ or $\lambda_C>0$ for some $C\in\mathcal{C}$ such that $k\in C$. And we remark that condition \eqref{eq:pcd} holding for all $k\in\mf$ is equivalent to the following:
for any $S\in\mathcal{C}-\{\mf\}$,
\begin{eqnarray}
\lefteqn{\triangle_S:=\nonumber}\\&& \sum_{C:C\subseteq S}\lambda_C\nonumber \\
&& -\frac{U_s+\sum_{C:C\not\subseteq S}\lambda_C\left(K-|C|+\frac{\mu}{\gamma}\right)}{1-\frac{\mu}{\gamma}} <0. \label{eq:repcd}
\end{eqnarray}
In particular, \eqref{eq:repcd} holds for all $S\in\mathcal{C}-\{\mf\}$ if it holds for all $S \in \{\mf-\{k\}:k\in\mf\}$.

\newpage
\section{Three Examples \label{sec:example}}

To illustrate Theorem \ref{thm:main}, we examine the three examples of P2P networks shown
in Figure \ref{fig:example}.

\begin{figure}[bth]
\subfigure[$K=1$]{
\post{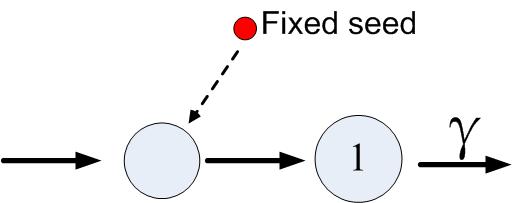}{5} \label{fig:K1}
}
\subfigure[$K=4$]{
\post{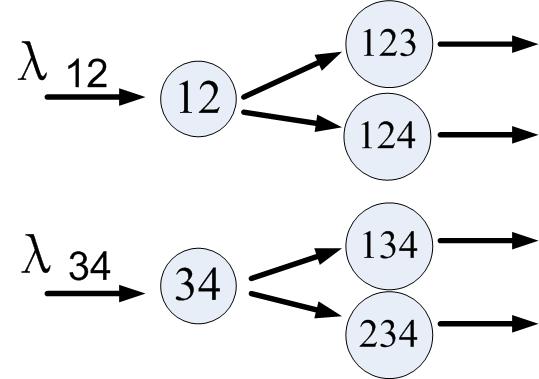}{5} \label{fig:K1234}}
\subfigure[$K=3$]{
\post{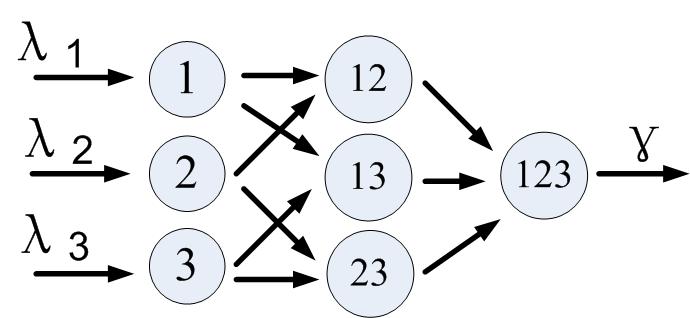}{7} \label{fig:K123}
}

\caption{Examples} \label{fig:example}
\end{figure}

{\bf Example $\bf 1$:} This example is treated in \cite{LeskelaRobertSimatos10}. As shown in Figure \ref{fig:K1},  the file is transferred as a single piece, that is, $K=1$. New peers without any piece arrive into the system at the  times of a Poisson process with rate $\lambda_0$. After obtaining the piece a peer becomes a peer seed. At rate $U_s$, the fixed seed contacts and uploads the piece to new peers, which become  peer seeds after obtaining the piece. When peer seeds are in the system, they randomly contact and upload copies of the piece to new peers with rate $\mu$, which creates more peer seeds. After staying for an exponentially distributed time period with mean $1/\gamma$, a peer seed leaves the system. This example illustrates  our model with parameters  $K=1$, $U_s,\mu,\gamma,\lambda_\emptyset = \lambda_0\in(0,\infty)$, and $ \lambda_{\{1\}} = 0$.

The stability of a system is determined by its ability to recover from a heavy load. First consider the case that there are many peer seeds in the system. Because every peer seed departs at rate $\gamma$, in essence, the service rate $\gamma x_{\mathcal{F}}$ scales linearly with the number of peer seeds, $x_{\mathcal{F}}$, as in an infinite server system, so the system can recover no matter how
 many peer seeds there are. Secondly consider the case that  there are many type $\emptyset$ peers and few peer seeds. For a long time period, when the fixed seed or a peer seed randomly contacts a peer to upload a piece, the probability they contact a type $\emptyset$ peer is close to one. So the group of type $\emptyset$ peers receives uploads from the fixed seed  at rate almost $U_s$. Once a peer becomes a peer seed, it can upload more pieces to type $\emptyset$ peers, creating more peer seeds, which upload more pieces.
So every peer seed can create a branching process of departures from the type $\emptyset$ group. 
The mean amount of time a peer seed stays in the system is $1/\gamma$, and during its stay it uploads pieces to type $\emptyset$ peers at rate close to $\mu$. So  on average, a peer seed can upload to $\mu/\gamma$ type $\emptyset$ peers. By the theory of branching process, if $\mu/\gamma\geq 1$, the expected number of descendants of a peer seed is infinite, which stabilizes the process. If $\mu/\gamma < 1$, on average every peer seed has $\mu/\gamma \over 1-\mu/\gamma$ descendants. Hence, every upload of the piece by the fixed seed to a  type $\emptyset$ peer causes, on average, about $1\over 1-\mu/\gamma$ departures from the type $\emptyset$ group. Comparing to $\lambda_0$, the arrival rate of type $\emptyset$ peers, this suggests that the system is stable if either $\mu\geq\gamma$, or $\mu<\gamma$ and $\lambda_0<U_s{1\over 1-\mu/\gamma}$. Conversely, if $\mu>\gamma$ and $\lambda_0>U_s{1\over 1-\mu/\gamma}$, the arrival rate of type $\emptyset$ peers is larger than the average rate of departures from the type $\emptyset$ group, indicating that the system cannot always recover from the heavy load of type $\emptyset$ group and so it is unstable. This conclusion is  confirmed by \cite{LeskelaRobertSimatos10} and Theorem \ref{thm:main}. 

{\bf Example $\bf 2$:} As shown in Figure \ref{fig:K1234}, the file is divided into four pieces, that is, $K=4$. There are two types of new peers, type $\{1,2\}$ and type $\{3,4\}$, which arrive as two independent Poisson processes with respective rates $\lambda_{12}$ and $\lambda_{34}$. There is no fixed seed in the system. Peers contact and upload pieces to each other so that they can depart. Peers depart immediately after obtaining all four pieces; there are no peer seeds in the system. This example illustrates our model with parameters  $K=4$, $U_s=0$, $\gamma=\infty$, $\mu,\lambda_{\{1,2\}} = \lambda_{12}, \lambda_{\{3,4\}} = \lambda_{34}\in(0,\infty)$, $\lambda_C=0$ for $C\neq \{1,2\},\{3,4\}$.

Consider the ability of the system to recover from a heavy load. First, consider the network starting from a state such that all peers are type $\{1,2,4\}$ and there are so many type $\{1,2,4\}$ peers that the fraction of them among all peers is close to one for a long time. On one hand, most new  type $\{1,2\}$ peers download piece $4$ from a type $\{1,2,4\}$ peer and join the type $\{1,2,4\}$ group, so the arrival rate of type $\{1,2,4\}$ peers is close to $\lambda_{12}$. On the other hand, most new  type $\{3,4\}$ peers download pieces $1$ and $2$ from type $\{1,2,4\}$ peers and then depart, with an expected lifetime in the system approximately ${2\over\mu}$. During its lifetime, a type $\{3,4\}$ peer uploads piece $3$ to two type $\{1,2,4\}$ peers on average and thereby
induces two departures on average. So the medium term aggregate departure rate of type $\{1,2,4\}$ peers is close to $2\lambda_{34}$. Hence, if $\lambda_{12}<2\lambda_{34}$, the system is able to recover from a heavy load of type $\{1,2,4\}$ (or $\{1,2,3\}$) peers. Conversely, if the inequality goes the other way, that is, $\lambda_{12}>2\lambda_{34}$, the arrival rate of type $\{1,2,4\}$ peers is larger than the aggregate departure rate of type $\{1,2,4\}$ peers. So the type $\{1,2,4\}$ group will keep growing. Thus if $\lambda_{12}>2\lambda_{34}$ the system cannot always recover from a heavy load of type $\{1,2,4\}$ (or $\{1,2,3\}$) peers. Similarly, if $\lambda_{34}<2\lambda_{12}$ the system can recover from a heavy load of type $\{2,3,4\}$ (or $\{1,3,4\}$) peers. And the system cannot always recover from the same heavy load if $\lambda_{34}>2\lambda_{12}$.

The situation is similar if there is a heavy load of type $\{1,2\}$ (or $\{3,4\}$) peers, while the other groups are empty. The arrival rate of type $\{1,2\}$ peers is $\lambda_{12}$. The aggregate departure rate of type $\{1,2\}$ peers, from the uploads of both type $\{3,4\}$ peers and
type $\{1,2,x\},x=3,4$ peers (which are former type $\{1,2\}$ peers), is larger than $2\lambda_{34}$. So if $\lambda_{12}<2\lambda_{34}$ the system is able to recover from the heavy load of type $\{1,2\}$ peers.

Secondly, consider the case that there are heavy loads in groups of at least two types, e.g. type $\{1,2\}$ and $\{1,2,3\}$. There is at least one type of peer that can upload to the other type of peer, e.g. type $\{1,2,3\}$ peers can upload to type $\{1,2\}$ peers. There are many uploads from type $\{1,2,3\}$ peers to type $\{1,2\}$ peers so that the departure rate from the type $\{1,2\}$ group is large, which stabilizes the system. This suggests that the system is stable if  $\lambda_{12}<2\lambda_{34}$ and $\lambda_{34}<2\lambda_{12}$, and unstable if either $\lambda_{12}>2\lambda_{34}$ or $\lambda_{34}>2\lambda_{12}$.  This conclusion is confirmed by Theorem \ref{thm:main}. 

{\bf Example $\bf 3$:} As shown in Figure \ref{fig:K123}, the file is divided into three pieces, that is, $K=3$. New peers arrive at a total rate $\lambda_{total}$, and each peer arrives with one piece, having piece $i$ with probability $\lambda_i/\lambda_{total}$. So there are three types of new peers, type $\{1\}$, type $\{2\}$, and type $\{3\}$, which arrive as three independent Poisson processes with rates $\lambda_1$, $\lambda_2$ and $\lambda_3$, respectively. There is no fixed seed in the system. At rate $\mu$ each, peers randomly contact and upload pieces to each other. After collecting all three pieces, every peer stays in the system as a peer seed for an exponentially distributed time with mean $1/\gamma, \gamma>\mu$. This example illustrates our model with parameters 
$K=3$, $U_s=0$, $0<\mu<\gamma\leq \infty$, $\lambda_{\{1\}}=\lambda_1, \lambda_{\{2\}} = \lambda_2, \lambda_{\{3\}} = \lambda_3\in(0,\infty)$, $\lambda_C = 0$ for $|C|\neq1$.

Consider whether the system can recover from a heavy load. First, consider the network starting from a state such that all peers are type $\{1,2\}$ and there are so many type $\{1,2\}$ peers that the fraction of them among all peers is close to one for a long time. By the reasoning of example two,  almost every new  type $\{1\}$ and type $\{2\}$ peer joins the type $\{1,2\}$ group, so the arrival rate of the type $\{1,2\}$ group is close to $\lambda_1+\lambda_2$. Over the medium term, every new  type $\{3\}$ peer has an expected lifetime approximately ${2\over\mu} +{1\over\gamma}$, with ${2\over\mu}$ being the expected time for the type $\{3\}$ peer to download two pieces from type $\{1,2\}$ peers, and with ${1\over\gamma}$ being the expected time for the type $\{3\}$ peer to be a peer seed. During its lifetime every type $\{3\}$ peer uploads approximately $2+{\mu\over\gamma}$ pieces to type $\{1,2\}$ peers on average. By the reasoning of example one, every peer seed creates a branching process of departures of type $\{1,2\}$ peers, with the total number of new peer seeds (including the root) equal to $1\over 1-\mu/\gamma$. Thus, on average, every new  type $\{3\}$ peer induces $(2+{\mu\over\gamma}){1\over 1-\mu/\gamma}$ departures from type $\{1,2\}$ group, so the medium term aggregate departure rate of type $\{1,2\}$ peers is approximately $\lambda_3(2+{\mu\over\gamma}){1\over 1-\mu/\gamma}$. Hence if $\lambda_1+\lambda_2<\lambda_3(2+{\mu\over\gamma}){1\over 1-\mu/\gamma}$, the system is able to recover from a heavy load of type $\{1,2\}$ group. Conversely, if $\lambda_1+\lambda_2>\lambda_3(2+{\mu\over\gamma}){1\over 1-\mu/\gamma}$, type $\{1,2\}$ group will keep increasing and the system cannot always recover from the heavy load. Similarly, if $\lambda_2+\lambda_3<\lambda_1(2+{\mu\over\gamma}){1\over 1-\mu/\gamma}$, or $\lambda_1+\lambda_3<\lambda_2(2+{\mu\over\gamma}){1\over 1-\mu/\gamma}$, the system is able to recover from a heavy load of type $\{2,3\}$, or $\{1,3\}$ group. And if either of the two inequalities is reversed, the system cannot always recover from a corresponding heavy load. 

Secondly, through considerations similar to those in example one and two, we can see that the conditions of heavy load in other single-type group or heavy load in multiple-type groups can also be recovered from if the three inequalities above hold. This suggests that the system is stable if
\begin{eqnarray*}
\begin{cases} \lambda_1+\lambda_2<\lambda_3(2+{\mu\over\gamma}){1\over 1-\mu/\gamma}\\ \lambda_2+\lambda_3<\lambda_1(2+{\mu\over\gamma}){1\over 1-\mu/\gamma}\\ \lambda_1+\lambda_3<\lambda_2(2+{\mu\over\gamma}){1\over 1-\mu/\gamma}
\end{cases}.
\end{eqnarray*} 
If any one of the three inequalities is reversed, it indicates the system is unstable. This is  consistent with Theorem \ref{thm:main}. Note that if peers depart immediately after obtaining a complete collection (i.e. $\gamma=\infty$), then the stability condition becomes
\begin{eqnarray*}
\begin{cases} \lambda_1+\lambda_2<2\lambda_3\\ \lambda_2+\lambda_3<2\lambda_1\\ \lambda_1+\lambda_3<2\lambda_2
\end{cases}.
\end{eqnarray*} 
If $\lambda_1,\lambda_2,\lambda_3$ are not all equal, at least one equality is reversed, so the system is unstable. This special case when $\gamma=\infty$ is considered in \cite{MassoulieVojnovic08}, and is discussed in Section \ref{sec:borderline} below.

\section{Outline of the Proof}  \label{sec:outline_of_proof}
The analysis of the above three examples suggests that when we consider the system to be in heavy load, the worst distribution of
load is that nearly all peers have the same type $C$ with $|C|=K-1$. If the system is able to recover from that kind of heavy load, it can recover from other kinds of heavy load. With this intuition in mind, a sketch of the proof of Theorem \ref{thm:main} is offered as follows.

First, we sketch the proof of Theorem \ref{thm:main}(a) about transience when $0<\mu<\gamma<\infty$. Without loss of generality, assume that \eqref{eq:qcd} is true for $k=1$, or equivalently, $\triangle_{\mf-\{1\}}>0$.

\begin{figure}
\post{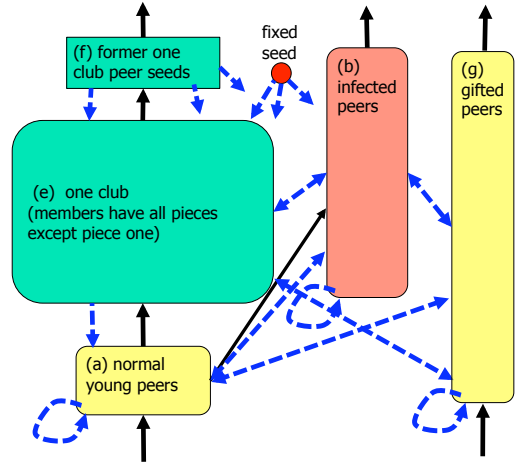}{10}
\caption{Flow of peers (solid lines) and pieces (dashed lines) in the system.   \label{fig:generation}}
\end{figure}

Consider the following partition of peers into five groups, as shown in Figure \ref{fig:generation}.
\begin{itemize}
\item{\em Normal young peer}: A normal young peer is a peer that is missing at least two pieces, one of them
being piece one.
\item{\em Infected peer}: An infected peer is a peer that obtained piece one after arriving, but before obtaining all the other pieces. Once a peer is infected, it remains infected until it leaves the system; it is considered to be infected even when it is a peer seed.
\item{\em Gifted peer}: A gifted peer is a peer that arrived with piece one. A gifted peer is gifted for its entire time in the system; it is considered to be gifted even when it is a peer seed.
\item{\em One-club peer}: A one-club peer is a peer that has all pieces except piece one. That is, the one-club is the group of peers of type $\{2,3,...K\}$.
\item{\em Former one-club peer}: A former one-club peer is a peer in the system that is not a one-club peer but at some earlier time was a one-club peer. Note that a former one-club peer is a peer seed. The converse is not true, because infected peers and gifted peers can be peer seeds.
\end{itemize}
Consider the system starting from an initial state in which there are many peers in the system, and all of them are one-club peers. The system evolves as shown in Figure \ref{fig:generation}. Piece one can arrive into the system from outside the system in two ways: uploads by the fixed seed or arrivals of gifted peers. Ignore for a second the effect of normal young peers getting piece one (and becoming infected). Most of the uploads by the fixed seed are uploads of piece one to one-club peers. One such upload creates a new peer seed, which on average will upload piece one to about $\mu/\gamma$ more one-club peers, and each of those will upload piece one to about $\mu/\gamma$ more one-club peers, and so forth, in a branching process. Each upload of a piece by the fixed seed thus ultimately causes, on average, about $1\over 1-\mu/\gamma$ departures from the one-club. Each gifted peer, with type $C$ on arrival, for some $C$ with $1\in C$, will directly upload to, on average, about $K-|C|+\mu/\gamma$ one-club peers, and those will become peer seeds which  also could upload to about $\mu/\gamma$ more one-club peers, and so fourth, so that the total expected number of one-club departures caused by the type $C$ gifted peer is $(K-|C|+\mu/\gamma){1\over 1-\mu/\gamma}$. Summing these quantities and subtracting them from the arrival rate of peers without piece one gives $\triangle_{\mathcal{F}-\{1\}}$. So $\triangle_{\mathcal{F}-\{1\}}>0$ indicates that the arrival rates of peers missing piece one is larger than the upload rate of piece one, causing the one-club size to grow linearly with time.

The above analysis neglects the possibility that normal young peers can also receive piece one, creating infected peers. An infected peer can upload to one club peers, creating former one-club peers, and to normal young peers, creating more infected peers. This results in a branching process comprised of infected peers and former one-club peers. By the theory of branching process, the expected number of infected offspring of a former one-club peer or an infected peer will converge to zero, as the fraction of one-club peers converges to one. Hence, when the one-club is large enough, the existence of infected peers does not appreciably affect the growth of the one-club. The detailed proof of transience is offered in Section \ref{sec:transience}.

Second, we sketch the proof of Theorem \ref{thm:main}(b) about positive recurrence for the case $0<\mu<\gamma< \infty$ under the assumption that $\eqref{eq:repcd}$ is valid for all $S\in\mathcal{C}-\{\mf\}$. The above discussion suggests that when $\triangle_{\mathcal{F}-\{1\}}<0$, the departure rate of the one-club is larger than the arrival rate of peers missing piece one, therefore, the system has the ability to recover from a single heavy load in the one-club. Moreover, when $k=2,3,...K$ and there is a single heavy load in the type $\mathcal{F}-\{k\}$ group, similar reasoning suggests that the system can recover
if $\triangle_{\mathcal{F}-\{k\}}<0$. 
To get a better idea of the proof, here we  consider other distributions of heavy load. 
\begin{itemize}
\item Suppose there is a single heavy load in some type $S$ group with $|S|\leq K-2$. Uploads from the fixed seed (with rate $U_s$) and from new  peers holding pieces not in $S$ (with rate $\sum_{C:C\not\subseteq S}\lambda_C$) keep creating departures from the type $S$ group. If we ignore the period of time from when a peer departs from the type $S$ group until the same peer becomes a peer seed, we see that the average remaining lifetime of every peer which departs from the type $S$ group is greater than or equal to
$1\over\gamma$. In this lifetime the peer uploads on average approximately $\mu/\gamma$ pieces to type $S$ peers, which creates more departures from the type $S$ group.  Including the root, every departure from the type $S$ group can ultimately cause at least $1\over 1-\mu/\gamma$ departures from the type $S$ group, on average. Because every new  type $C$ peer with $C\not\subseteq S$ eventually uploads on average $K-|C|+\mu/\gamma$ pieces to type $S$ peers, the departure rate of type $S$ group is larger than $\left[U_s+\sum_{C:C\not\subseteq S}\lambda_C(K-|C|+\mu/\gamma)\right]{1\over 1-\mu/\gamma}$. Because  peers mainly download pieces from type $S$ peers, almost all new type $C$ peers with $C\subseteq S$ ultimately join the type $S$ group. So the near term arrival rate of type $S$ group is less than but close to $\sum_{C:C\subseteq S}\lambda_C$, which is smaller than the aggregate departure rate of type $S$ peers by \eqref{eq:repcd}. So the system can recover from the heavy load. 
\item Suppose there is a single heavy load in the type $\mathcal{F}$ group, that is, the group of peer seeds. The departure rate of peer seeds, $\gamma x_{\mathcal{F}}$, scales linearly with the number of peer seeds, $x_{\mathcal{F}}$, as in an infinite server
queueing system. So the system can recover however large the group of peer seeds is.
\item Suppose there are heavy loads in at least two groups of different types, say types $C_1$ and $C_2$. In this condition, either $C_1\not\subsetneq C_2$ or $C_2\not\subsetneq C_1$ is true, so peers in at least one of the groups, say $C_1$, can upload pieces to peers in the other group, say $C_2$. The  rate of peers departing from the type $C_2$ group is quite high, due to the large rate of uploads from type $C_1$ peers, so the system can quickly escape from that region of the state space. 
\end{itemize}

The above paragraphs summarize how the system can recover from all distributions of heavy load. To provide a proof of stability it must also be shown that the load cannot spiral up without bound through some oscillatory behavior. For that we use a Lyapunov function and apply the Foster-Lyapunov stability criterion. The detailed proof is offered in Section \ref{sec:positiverecurrence}.

\section{Proof of Transience in Theorem \ref{thm:main}\label{sec:transience}}

In the following the detailed proof of Theorem \ref{thm:main}(a) is given. It is obvious the system is transient if no copies of piece $k$ can enter the system.  Without loss of generality, assume $0<\mu<\gamma\leq\infty$ and assume $\triangle_{\mathcal{F}-\{1\}}>0$. For a given time $t\geq 0,$ 
define the following random variables, using the terminology of Section \ref{sec:outline_of_proof} and Figure \ref{fig:generation}:
\begin{itemize}
\item $Y_t^a:$ number of normal young peers (group (a)) at time $t$.
\item $Y_t^b:$ number of infected peers (group (b)) at time $t$.
\item $Y_t^g:$ number of gifted peers (group (g)) at time $t$.
\item $Y_t^e:$ number of one-club peers (group(e)) at time $t$.
\item $Y_t^f:$ number of former one-club peers (group (f)) at time $t$.
\item $A_t:$ cumulative number of arrivals, up to time $t$, of peers without piece one at time of arrival
\item $D_t:$ cumulative number of downloads of piece one, up to time $t.$ (Peers arriving with piece one are not counted.)
\item $N_t:$ number of peers at time $t$.
\end{itemize}

The system is modeled by an irreducible, countable-state Markov chain. A property of such random processes is
that either all states are transient, or no state is transient.  Therefore, to prove Theorem \ref{thm:main}(a),
it is sufficient to prove that some particular state is transient.    With that in mind, we assume that the initial state is the one with
$N_o$ peers, and all of them are one-club peers, where $N_o$ is a large constant specified below.
Given a small number $\xi$ with $0<\xi < 1,$  let $\tau$ be the extended
stopping time defined by  $\tau=\min\{t \geq 0 : Y^e_t+Y^f_t   \leq   (1-\xi) N_t  \},$ with the usual convention that $\tau=\infty$ if  
 $Y^e_t+Y^f_t  >  (1- \xi) N_t$ for all $t.$
 It suffices to prove that
\begin{equation}  \label{eq:target}
P\{ \tau = \infty ~\mbox{and} \lim_{t\rightarrow\infty} N_t=+\infty \} > 0.
\end{equation}

The probability of the event in \eqref{eq:target} depends on only the out-going transition rates for states
such that $Y^e +Y^f  >  (1-\xi) N .$    Thus, we can and do
prove \eqref{eq:target} instead for an alternative system, that has the same initial state, and the same
out-going transition rates for all states such that  $Y^e +Y^f  >   (1-\xi)   N,$   as the original system.   
The alternative system, however, guarantees an upper bound on the aggregate rate of downloads of piece
one by peers in group (a), and  a lower bound on the rate of downloads by the set of peers
in groups (a), (e) and (f).   This can be done so that the alternative system has the  following six
properties, the first four of which hold for the original system, and the last two of which hold for the original
system on the states with $Y^e +Y^f  >   (1-\xi)   N:$
\begin{enumerate}
\item A peer with a complete collection departs according to an exponentially distributed random variable
with parameter $\gamma.$
\item Each peer in group  (b), (g), or (f) uploads to the set of peers in group (e) with rate at most $\mu.$
\item The fixed seed  uploads to the set of peers in group (e) with rate at most $U_s.$
\item  Peers in group (a) can download piece one only from peers in groups (b), (g), or (f), or from the fixed seed.
\item  Whenever the internal Poisson clock of a peer in group (b), (g), or (f), or the fixed seed, ticks, the 
probability the tick results in contacting a peer in group (a) to upload to is less than or equal to $\xi.$
\item A peer in group (a), (b), or (g) that is not yet a seed peer receives usable download opportunities at rate
greater than or equal to $(1-\xi)\mu.$  (If $Y^e +Y^f   > (1- \xi )N,$  these opportunities can be provided
by the peers in groups (e) and (f). )
\end{enumerate}
The alternative system can be defined by supposing that on the
states with $Y^e +Y^f  \leq    (1-\xi)   N:$  (i)  the opportunities for peers in groups (b), (g) or (f), or the fixed seed, to
download to peers in group (a) are discarded with
some state-dependent positive probability, and (ii)  there is a phantom seed, having all pieces except piece one, that
uploads pieces to peers in groups (a), (b), or (g) as necessary for property 6) above to hold.  
{\em For the remainder of this proof we consider the alternative
system, but for brevity of notation, use the same notation for it as for the original system, and refer to it as the original system.}


Only peers in groups (a) and (e) download piece one;
peers in the other three groups already have  piece one.  
A peer in group (a) downloading piece one immediately moves to group (b), and
a peer in group (e) downloading piece one immediately moves to group (f).
Thus, a download of piece one creates either a group (b) peer or
a group (f) peer.  A group (b) peer or group (f) peer stays in the same group until
it leaves the system.  While a peer in group (b) or (f) is in the system
it can generate more peers in groups (b) and (f) by uploading piece one, and those
peers are considered to be offspring spawned by the peer.      Since offspring can themselves
spawn offspring,  there is a branching process, and a group (b) or group (f) peer
has a set of descendants.


We shall consider the evolution of a portion of the system under some statistical assumptions
that are different from those in the original system.  We refer to it as the
{\em autonomous branching system} (ABS) because strong independence
assumptions are imposed.    The ABS
pertains only to those peers that have piece one.   It is shown below that the original
system can be stochastically coupled to the ABS so that uploads of
piece one happen in the original system only when they also happen in the ABS. 
We begin by considering only group  (b) and group (f) peers.  
In the original system,  a group (b) peer was formerly a group (a) peer, and a group (f) peer
was formerly a group (e) peer; such previous history is irrelevant for the system under the
ABS; the description below concerns such a peer only from the
time it becomes a group (b) or group (f) peer.   The statistical assumptions for the ABS
 involving these peers are as follows:
\begin{itemize}
\item A group (b) peer is required to download $K-1$ pieces;
usable opportunities for such downloads arrive according to a Poisson
process of rate $\mu(1-\xi).$   (The interpretation is that, when a group (b) peer appears,
any piece it might have had besides piece one is ignored or discarded.) 
After the $K-1$ downloads, the group (b) peer remains in the system as a seed peer
for a {\em seed dwell duration} that is exponentially distributed with parameter $\gamma.$
\item A group (f) peer remains in the  system for a seed dwell duration that is
exponentially distributed with parameter $\gamma.$
\item A group (b) peer or group (f) peer spawns group (b) peers according to a
Poisson process of rate $\xi\mu$ and  it spawns group (f) peers according to a
Poisson process of rate $\mu.$
\item The Poisson processes for spawning offspring, as well as the seed dwell durations,
are mutually independent.
\end{itemize}
The above assumptions uniquely determine the distribution of the number
of offspring, and therefore the total number of descendants, of a group (b)
or group (f) peer.   On average, a group (b) peer is in the system (as a group (b) peer)
for $\frac{K-1}{(1-\xi)\mu} +  \frac{1}{\gamma}$ time units, and thus on
average a group (b) peer spawns 
$\xi (\frac{K-1}{1-\xi} +  \frac{\mu}{\gamma})$
 offspring of type
(b) and $\frac{K-1}{1-\xi} +  \frac{\mu}{\gamma}$ offspring in group (f).   Similarly, on average, 
a peer in group (f) spawns  $ \frac{\xi\mu}{\gamma}$  offspring of type
(b) and $ \frac{\mu}{\gamma}$ offspring in group (f).  
Let $m_b$ denote one plus the mean number of descendants of a group (b) peer and
let $m_f$ denote one plus the mean number of descendants of a group (f) peer, in the ABS.
Then by the theory of branching processes, 
${m_b \choose m_f}$ is the minimum nonnegative solution to the equations
$$
{m_b \choose m_f} = {1 \choose 1} + \left( \begin{array}{cc}
  \xi\left(\frac{K-1}{1-\xi} +\frac{\mu}{\gamma}\right)  & \frac{K-1}{1-\xi} +\frac{\mu}{\gamma} \\
\frac{\xi\mu}{\gamma}  & \frac{\mu}{\gamma}  \end{array} \right) {m_b \choose m_f}.
$$
The two-by-two matrix involved here has rank one, and the solution is easily found to be finite
if
\begin{equation}  \label{eq.stability.xi}
\xi\left(\frac{K-1}{1-\xi} +\frac{\mu}{\gamma}\right)  + \frac{\mu}{\gamma}  < 1.
\end{equation}
If \eqref{eq.stability.xi} holds,
$$
{m_b \choose m_f} = {1 \choose 1} + \frac{1+\xi}{1-\xi(\frac{K-1}{1-\xi} +\frac{\mu}{\gamma})-\frac{\mu}{\gamma}}{\frac{K-1}{1-\xi} +\frac{\mu}{\gamma}  \choose \frac{\mu}{\gamma}},
$$
and, in addition, the second moment of the number of descendants of a peer of either group (b) or (f) is finite and
monotonically increasing in $\xi.$
Note that
$$
{m_b \choose m_f}\stackrel{\xi\rightarrow 0}{\rightarrow}
\left( \begin{array}{c}   \frac{K}{1-\frac{\mu}{\gamma}}   \\   \frac{1}{1-\frac{\mu}{\gamma}}    \end{array}\right).
$$

Next, we extend the scope of the ABS
to include a gifted peer; this entails the following statistical assumptions:
\begin{itemize}
\item A gifted peer with piece collection $C$ upon arrival is required to download $K-|C|$
pieces; usable opportunities for such downloads arrive according to a Poisson process
of rate $\mu(1-\xi).$    After the $K-|C|$ downloads, the group (b) peer remains in the
system as a seed peer for a seed dwell duration that is exponentially distributed with
parameter $\gamma.$

\item While a gifted peer is in the system, it spawns group (b) peers according
to a Poisson process of rate $\xi\mu$ and it spawns group (f) peers according to a
Poisson process of rate $\mu.$
\item The Poisson processes for spawning offspring, as well as the seed dwell duration,
are mutually independent.
\end{itemize}
The mean time a gifted peer with initial piece collection $C$ 
is in the system is thus $\frac{K-|C|}{\mu(1-\xi)} + \frac{1}{\gamma},$
so the mean total number of descendants of a gifted peer ({\em not} including the gifted peer itself)
is given by
\ben
m_g(C)  & = &  \left(  \frac{K-|C|}{\mu(1-\xi)} + \frac{1}{\gamma}\right)  (\xi \mu m_b + \mu m_f)  \\
&=&  \left( \frac{ K-|C| }{1-\xi} + \frac{\mu}{\gamma}   \right)  (\xi  m_b + m_f).
\een
Note that  $ m_g(C)  \stackrel{\xi\rightarrow 0}{\rightarrow}    \left(  K-|C| + \frac{\mu}{\gamma}   \right) \frac{1}{1-\frac{\mu}{\gamma}}.$

Finally, we extend the scope of the ABS
to include the processes of arrivals of gifted peers and the uploads of the fixed seed; this
entails the following assumptions:

\begin{itemize}
\item For each $C$ with $1\in C$, gifted peers with initial piece collection $C$
arrive according to a Poisson process of rate $\lambda_C$
(as in the original model).
\item The fixed seed spawns peers in group (b) according to a Poisson process of rate $\xi U_s$
and it spawns peers in group (f) according to a Poisson process of rate $U_s.$
\item The Poisson processes of arrivals are mutually independent.
\item  Gifted peers and offspring of the fixed seed are considered to be {\em root peers.}  The
evolution of the descendants of root peers are mutually independent.
\end{itemize}

Let $\widehat{D}_t$ denote the cumulative number of group (b) and group (f) peers appearing
in the ABS up to time $t.$

\begin{lemma}  \label{lemma.onecompare}
The process $(D_t: t\geq 0)$ is stochastically dominated (see the appendix for
the definition) by $(\hD_t: t\geq 0).$
\end{lemma}

\begin{proof}
We  describe a particular method of coupling the ABS and the original system. 
By this, we mean a way to construct both processes on a single probability space.
To do this, we start with the random variables governing the ABS, and
describe how the original system (i.e. a system with the statistical description of the
original system) can be overlaid on the same probability space, in
such a way that $D_t \leq \hD_t$ for all $t$ with probability one.

The first step is to adopt a new way of thinking about the ABS.  In the ABS,
the sets of descendants of the root peers form a partition of all group (b) and
group (f) peers in the ABS (for this purpose, the descendants of a root peer
include the root peer itself if the root peer is an offspring of the fixed seed,
but not if the root peer is a gifted peer).   Imagine that each root peer arrives with a
randomly generated {\em script} for itself and its descendants.  The script includes the sample paths
of the Poisson processes that determine: when pieces are to be downloaded,
when group (b) peers are spawned, and when group (f) peers are spawned, as
well as the seed dwell durations sampled from the exponential distribution with
parameter $\gamma.$   Whenever some peer in the ABS system spawns another,
the portion of the script held by the parent associated with that offspring and
its descendants becomes a script for that offspring.

The next step is to build the original system using the same random variables, using the following
assumptions.  When thinning of Poisson processes is mentioned, it refers to randomly rejecting
some points of a Poisson process to produce another point process with a specified intensity that
is smaller than the rate of the Poisson process.
\begin{enumerate}
\item  The original system has independent Poisson arrivals of peers of type $C$ at rate $\lambda_C,$
for all $C$ with $1\not\in C.$  These arrivals are not modeled in the ABS, and are to be
generated for the original system independently of the ABS.
\item   The arrival processes of gifted peers of type $C$ in the original system for all $C$
with $1\in C$ are identical to those in the ABS.
\item   The point process of times that the seed uploads piece one to one-club peers
is a thinning of the rate $U_s$ Poisson process governing creation of group (f)
peers in the ABS system.
\item   The point process of times that the seed uploads piece one to normal young peers
is a thinning of the rate $\xi U_s$ Poisson process governing creation of group (b)
peers in the ABS system.  
\item   The point process of times that a peer in group (b), (g), or (f) uploads piece one to
one-club peers is a thinning of the rate $\mu$ Poisson process in the script of the peer
for spawning group (f) peers.   
\item   The point process of times that a peer in group (b), (g), or (f) uploads piece one to
normal young peers is a thinning of the rate $\xi \mu$ Poisson process in the script of the peer
for spawning group (b) peers.   
\item A peer in the original system in group (b) or (g) that does not have a complete collection,
downloads useful pieces from a peer in group (e) or (f) at the jump times of the rate $\mu(1-\xi)$ Poisson
process for downloads in its script.  The peer can also make downloads at other times, to bring the
total intensity of downloads from groups (e) and (f) up to at least $ \frac{\mu(Y^e+Y^f)}{N}.$
\item The peer seed dwell time for any peer is specified in its script.
\end{enumerate}

A remark is in order about why the construction is possible.
When one peer transfers a piece to another peer in the original system, it is considered an
upload for the first peer and a download for the second.   Thus, the timing of such transfers
cannot be simultaneously governed by internal scripts of the two peers.    In the construction noted here,
such conflict does not occur, because the scripts are used to determine times that piece one can be uploaded,
and the peers that are downloading piece one are in group (a) or (e), and are thus not yet following
a script.   And the scripts are used for downloading of pieces other than piece one, but do not constrain
times that pieces other than piece one are uploaded.

The resulting coupling satisfies the following properties.
\begin{itemize}
\item Any peer in group (b), (f), or (g) in the original system is also in the ABS, in the same group and with the same
time of arrival to that group.    (Such peers can remain in the ABS longer than they stay in the original system.)
\item Any peer in group (f) in the original system (and thus also in the ABS)  departs from both systems
at the same time.   (Peers in groups (b) or (g) in the original system can stay longer in the ABS than in the original
system.)
\item Whenever some peer $p_1$ in the original system uploads piece one to some other peer $p_2$, 
peer $p_1$ simultaneously spawns peer $p_2$ in the ABS.  Afterwards, peer $p_2$ is either in group (b)
or in group (f) in both systems.
\item There can be more group (b) and more group (f) peers in the ABS than in the original system because the
spawning rates in the ABS system are greater than in the original system, and group (b) and group (f) peers in
the ABS can have fewer pieces in the ABS system than they have in the original system, and thus they can stay
longer in the ABS system than in the original system.
\end{itemize}

In particular, by the third point above, whenever piece one is uploaded in the original system a peer of group (b) or (f)
is created in the ABS system.   Therefore, $D_t \leq \hD_t$ for all $t\geq 0$ with probability one, which by the definition
of stochastic domination, proves the lemma.
\end{proof}

\begin{corollary}  \label{cor.D}
Given $\epsilon > 0,$  if $\xi$ is sufficiently small, then for all $B$ sufficiently large, 
\begin{equation}
P\left\{ 
\begin{array}{c} 
D_t  <  B  + ~~~~~~~~~~~ \\ 
  \frac{   U_s +  \sum_{C: 1\in C} \lambda_C\left(K-|C| +\frac{\mu}{\gamma} \right) }{ 1-\frac{\mu}{\gamma}}t + \epsilon t    \\   
  \mbox{for all}~ t\geq 0   
 \end{array}  \right\} \geq 0.9.  \label{eq.Dineq} 
\end{equation}
\end{corollary}

\begin{proof}
By Lemma \ref{lemma.onecompare}, it suffices to prove Corollary \ref{cor.D} with $D$
replaced by $\hD.$
Let $\hhD$ be a random process associated with the ABS,  denoting the cumulative counting
process that results if all the descendants of a root peer are counted at the time the root peer arrives. 
The, processes $\hD$ and $\hhD$ count the same downloads of piece one, but $\hhD$
does so sooner, so $\hD_t \leq \hhD_t$ for all $t.$  Thus, it suffices to prove Corollary \ref{cor.D} with
$D$ replaced by $\hhD.$   The process $\hhD$ is a compound Poisson process,
which can be decomposed into the sum of several independent compound Poisson processes: one for
each type $C$ with $1\in C,$  and one for peer seeds generated directly by the fixed seed.
The mean arrival rate for $\hhD$ satisfies:
\begin{eqnarray*}
 \frac{E\left[\hhD_t\right] }{t} & = &  U_s (\xi  m_b + m_f)  +  \sum_{C: 1\in C}  \lambda_C m_g(C)  \\
& \stackrel{\xi \rightarrow 0}{\rightarrow}  &
\frac{ U_s   +\sum_{C: k\in C} \lambda_C\left(K -|C|  +\frac{\mu}{\gamma} \right) }{1-\frac{\mu}{\gamma}}.
\end{eqnarray*}
and the batch sizes have finite second moments for $\xi$ sufficiently small, and the second moments are increasing in $\xi.$
Therefore, Corollary \ref{cor.D} follows from Kingman's moment bound (see Proposition \ref{prop:compoundKingman} in the appendix.)
\end{proof}

\begin{lemma}
Given $\epsilon > 0,$  if $B$ is sufficiently large,
\begin{equation}
P\left\{ A_t >   -B+  \left(  \sum_{C:k\not\in C}\lambda_C  -  \epsilon\right) t .~\forall t\geq 0\right\}    \geq 0.9. \label{eq.Aineq}  
\end{equation}
\end{lemma}
\begin{proof}
The process $A$ is a Poisson process with rate $ \sum_{C:k\not\in C}\lambda_C.$
Thus,  \eqref{eq.Aineq} follows from Kingman's moment bound
(see Proposition \ref{prop:compoundKingman} in the appendix.)

\end{proof}

\begin{lemma}   \label{lemma.MG1compare}
The process $Y^a_t+Y^b_t+Y^g_t$ is stochastically dominated by the number of customers in an $M/GI/\infty$ queueing system
with initial state zero,  arrival rate $\lambda_{total},$ and service times having mean $\frac{K}{\mu(1-\xi)}+\frac{1}{\gamma}.$
\end{lemma}

\begin{proof}
The idea of the proof is to show how, with a possible enlargement of the
underlying probability space, an $M/GI/\infty$  system can be constructed on the same probability
space as the original system, so that for any time $t$, $Y^a_t+Y^b_t+Y^g_t$ is
less than or equal to the number of peers in the $M/GI/\infty$ system.   Let the $M/GI/\infty$ system have
the same arrival process as the original system--it is a Poisson process of rate $\lambda_{total}.$

An important point is that any peer in group (a), (b), or (g) is either receiving useful download
opportunities at rate at least $(1-\xi)\mu,$  or is a peer seed (possible if it is in group (b) or (g)) and
is thus waiting for a departure time that is exponentially distributed with parameter $\gamma.$
We can thus imagine that any arriving peer has an internal Poisson clock that ticks at
rate $\mu(1-\xi)$ and an internal, exponentially distributed random variable with parameter
$\gamma.$   Whenever its internal clock ticks, it can download a useful piece, until it either
joins the one club (in which case it leaves group (a) and joins group (e)) or it becomes a peer
seed, in which case it remains in the system as a peer seed for an amount of time equal to
its internal exponential random variable of parameter $\gamma.$

An arriving peer in the original system may already have some pieces at the time of arrival, or its
 intensity of downloading pieces could be greater than $(1-\xi)\mu,$  or it might leave the
 union of groups (a), (b), and (g)  by becoming a one-club peer.  These factors cause to reduce the time
 that a peer remains in the union of groups  (a), (b) and (g).    The  $M/GI/\infty$ system system is constructed by
 ignoring those speedup factors.    Specifically,  in the $M/GI/\infty$ system,
each arriving peer has to download $K$ pieces at times governed by its internal Poisson clock, and then remain
as a peer seed for a time duration given by its internal exponentially distributed random variable
for seed time.   The service time distribution for the $M/GI/\infty$ system is thus the sum of
$K$ independent exponential random variables with parameter $\mu(1-\xi)$ plus a single exponential
random variable with parameter $\gamma.$   Any peer that is in groups (a), (b), or (g) in the
original system will be in the $M/GI/\infty$ system, and the mean service time for the
$M/GI/\infty$ system is $\frac{K}{\mu(1-\xi)}+\frac{1}{\gamma}.$
\end{proof}

\begin{corollary}  Given $\epsilon_o > 0 $ and $\xi > 0$,    if $B$ is sufficiently large,
\begin{eqnarray}
P\{ Y^a_t  +Y^b_t+Y^g_t  <   B  +  \epsilon_o t ~~ \mbox{for all}~ t\geq 0  \}    \geq 0.9.   \label{eq.Yineq}  
\end{eqnarray}
\end{corollary} 
\begin{proof}
The Corollary follows from Lemmas  \ref{lemma.MG1compare} and
  \ref{lemma.mginfty} with $m$ in Lemma  \ref{lemma.mginfty} equal to
$\frac{K}{\mu(1-\xi)}+\frac{1}{\gamma}$ and $\epsilon$ equal to
$\epsilon_o.$ 
\end{proof}

\vspace{.1in}
The proof of Theorem \ref{thm:main}(a)  is now completed.
\begin{itemize}
\item Select $\epsilon > 0$ so that $2\epsilon < \triangle_{\mathcal{F}-\{1\}}.$ 
\item Select $\xi>0$ so small that \eqref{eq.Dineq} holds for sufficiently large $B.$
\item Select $\epsilon_o$ small enough that
$
\frac{\epsilon_o}{ \triangle_{\mathcal{F}-\{1\}}-2\epsilon} < \xi.
$
\item Select $B$ large enough that \eqref{eq.Dineq}, \eqref{eq.Aineq}, and \eqref{eq.Yineq} hold.
\item  Select $N_o$ large enough that  $\frac{B}{N_o-2B}\leq \xi$.
\end{itemize}
Let $\cal E$ be the intersection of the three events on the left sides of  \eqref{eq.Dineq},
\eqref{eq.Aineq}, and \eqref{eq.Yineq}.  By the choices of the constants,
 \eqref{eq.Dineq}, \eqref{eq.Aineq}, and \eqref{eq.Yineq} hold, so that
 $P\{{\cal E} \} \geq 0.7.$   To complete the proof, it will be shown that
$\cal E$ is a subset of the event in \eqref{eq:target}, thereby establishing  \eqref{eq:target}.
Since $N_t$ is greater than
or equal to the number of peers in the system that don't have piece one,
on $\cal E$, \\ $N_t \geq N_o+A_t -D_t  >  N_o - 2B + (\triangle_{\mathcal{F}-\{1\}}-2\epsilon) t $ for all $t\geq 0.$    Therefore, on $\cal E$, for
any $t\geq 0,$
\begin{eqnarray*}
\lefteqn{\frac{Y^a_t  +Y^b_t+Y^g_t }{N_t}  }\\&  <    & \frac{ B  +  \epsilon_o t  }{ N_o  - 2B + (\triangle_{\mathcal{F}-\{1\}}-2\epsilon) t }  \\
&\leq & \max \left\{  
\frac{ B   }{ N_o -2B} , ~
 \frac{  \epsilon_o   }{\triangle_{\mathcal{F}-\{1\}} -2\epsilon} 
\right\}  \leq \xi.
\end{eqnarray*}
Thus,  $\cal E$ is a subset of the event in \eqref{eq:target} as claimed.
This completes the proof of Theorem \ref{thm:main}(a).

\section{Proof of Positive Recurrence in Theorem \ref{thm:main} \label{sec:positiverecurrence}}
Theorem \ref{thm:main}(b) is proved in this section. The first subsection treats the case $0<\mu<\gamma\leq\infty$ and the second subsection treats the case $0<\gamma\leq \mu$.

\subsection{Proof of Positive Recurrence when $0<\mu<\gamma\leq \infty$ in Theorem \ref{thm:main}\label{sec:p2}}

The detailed proof of Theorem \ref{thm:main}(b) when $0<\mu<\gamma\leq \infty$ is given in this subsection. Assume $0<\mu<\gamma\leq \infty$ and assume \eqref{eq:repcd} is valid for all $S\in\mathcal{C}-\{\mf\}$. 
For any nonnegative function $F=F(\mathbf{x})$ on the state space of the system, the drift of $F$ at state $\mathbf{x}$ is defined as 
\begin{eqnarray}
Q(F)(\mathbf{x}): = \sum_{\mathbf{x}':\mathbf{x}'\neq \mathbf{x}} q(\mathbf{x}, \mathbf{x}') \left[F(\mathbf{x}')-F(\mathbf{x})\right]. \label{eq:drift}
\end{eqnarray}
If, as usual, the diagonal elements $q(\mathbf{x},\mathbf{x})$ of the transition matrix $Q$ are chosen so that row sums are zero, $Q(F)$ is the product of the matrix $Q$ and function $F$, viewed as a vector.   In this paper, we apply the following
lemma implied by the Foster-Lypunov criterion.
\begin{lemma} \label{prop:FL}
The P2P Markov process is positive recurrent and $E[N]<+\infty,$  where $N$ is a random variable with the stationary
distribution for the number of peers in the system, if there is a nonnegative function $W(\mathbf{x})$ on the state space of the process, such that (i)  $\{\mathbf{x}:W(\mathbf{x})\leq c\}$ is a finite set for any constant $c$, and
(ii)  there exists $n_o\geq 0$ and $ \xi>0$ so that $QW\leq -\xi n <0$ whenever $n\geq n_o.$ 
We call such a $W$ a valid Lyapunov function.
\end{lemma}
\begin{proof}
For any $\mathbf{x}$, $q(\mathbf{x},\mathbf{x}')$ is nonzero for only finitely many values of $\mathbf{x}'$, so $QW$ is finite for all $\mathbf{x}$. 
Therefore the constant
$\widehat{B},$ defined by  $\widehat{B}=\max_{\mathbf{x}:n<n_o} QW(\mathbf{x}),$ is finite.  The lemma follows from the combined Foster-Lyapunov
stability criterion and moment bound--see Proposition \ref{cor.FosterCompContinuous} in the appendix--with
$V=W,$ $f(\mathbf{x})=\xi n,$  and $g(\mathbf{x})=\widehat{B}I_{\{n< n_o\}}.$
\end{proof}

The Lyapunov function we  use is,  if $0<\mu<\gamma<\infty:$
\begin{equation}
W:=   \sum_{C:C\in \mathcal{C}} r^{|C|} T_C,  \label{eq:defW}
\end{equation}
where
$$
 T_C :=
\begin{cases}
{1\over2}E_C^2 + \alpha E_C\phi(H_C)& \text{ if }C\neq \mathcal{F}\\
{1\over2}n^2&\text{ if }C=\mathcal{F}
\end{cases},
$$
and if $0<\mu<\gamma=\infty:$
\be
W &: =& \sum_{C:C\in \mathcal{C}-\{\mf\}} r^{|C|} T_C \label{eq:defW2},
\ee
with the following notation:
\begin{itemize}
\item $r\in(0,{1\over2}), d\in(1,\infty),\beta\in (0,{1\over2}), \alpha\in({1\over2},1)$ are positive constants to be specified, with $r$ and $\beta$ small, $d$ large, and $\alpha$ close to one.
\item$\mE_C: = \{C':C'\subseteq C\}$, which is the collection of types of peers which are or can become type $C$ peers.
Note that $\mE_C$ is downward closed (i.e. a lower set) for any $C.$
\item $\mH_C: = \{C': C'\in\mathcal{C},C'\not\subseteq C\}$,
which is the collection of types of peers which can help type $C$ peers. 
Note that $\mE_C$ is upward closed (i.e. an upper set) for any $C.$
Also,  $\mf\in\mH_C$ for any $C\in\mathcal{C}-\mf$ and $\mH_\mf = \emptyset$.
\item$E_C: = \sum_{C':C'\in\mE_C}x_{C'}$,  e.g. $E_\mf = n.$
\item $H_C: = {1\over 1-\mu/\gamma} \sum_{C':C'\in\mH_C} (K-|C'|+\mu/\gamma) x_{C'},$ e.g. $H_\mf = 0.$
\item$\phi$ is the function with parameters $d,\beta$, defined as 
\begin{eqnarray*}
\phi(x): = \begin{cases}
(2d+{1\over 2\beta}-x)& \text{if } 0\leq x\leq 2d\\
{{\beta\over2}(x-2d-{1\over\beta})^2}& \text{if }2d<x\leq 2d+{1\over \beta}\\
0 &\text{if }x>2d+{1\over \beta}
\end{cases}.
\end{eqnarray*}
Thus $\phi'(x)=-1$ for $0\leq x\leq 2d$, $\phi'(x)=0$ for $x\geq 2d+1/\beta$, and $\phi'$ increases linearly from $-1$ to $0$ over the interval $[2d,2d+1/\beta]$. In particular, $-1\leq \phi'(x)\leq 0$ for $x\geq 0$.
\end{itemize}

In the proof, we consider the following two classes of states, where $\epsilon$ is
to be selected with $0 < \epsilon < {1\over2}. $  The classes overlap and their
union includes every nonzero state: 
\begin{definition} \label{def:class}
Class \Rmnum{1} is the set of states $\mathbf{x}$ such that  there exists $S\in\mathcal{C}-\{\mf\}$, so that $x_S/n>1-\epsilon$; class \Rmnum{2} is the set of states $\mathbf{x}$ such that there exist $C_1, C_2\in \mathcal{C}$, either $C_1$ and $C_2$ being distinct or both equal to $\mathcal{F}$, so that, $x_{C_1}/n>\epsilon/2^K$ and $ x_{C_2}/n>\epsilon/2^K$. 
\end{definition}

The main idea of the proof is to show that  $W$ is a valid Lyapunov function
for an appropriate choice of $(r,d,\beta,\alpha, \epsilon).$
The given parameters of the network, $K,U_s,\lambda=\left(\lambda_S:S\in\mathcal{C}\right),\gamma$ and $\mu$, are treated as constants. Functions on the state space are considered which may depend on
the variables $r, d,\beta,\alpha$ and $\epsilon.$  It is convenient to adopt the big theta notation $\Theta(\ast)$, with
the understanding that it is uniform in these variables; this is summarized in the following definitions. 
\begin{definition}
Given functions $f$ and $g$ on the state space, we
say $f = \Theta(g)$ if there exist constants $k_1,k_2, n_o>0$, not depending on $(r,d,\beta,\alpha, \epsilon)$, such that $ k_1|g(\mathbf{x})|\leq |f(\mathbf{x})| \leq k_2 |g(\mathbf{x})|$ for all $\mathbf{x}$ such that $n>n_o$.
\end{definition}

For example, $2=\Theta(1), \lambda_{total} n=\Theta(n)$, $d\not=\Theta(1)$, $d=\Theta(d)$. Similarly, we adopt notions of ``small enough" and ``large enough" that are uniform in $(r,d,\beta,\alpha,\epsilon)$:
\begin{definition}
The statement, ``condition $A$ is true if $x>0$ is {\em small enough}", means there exists a constant $ k>0$, not depending on $(r,d,\beta,\alpha,\epsilon)$, such that $A$ is true for any $x\in(0,k)$. Similarly, the statement, ``condition $A$ is true if $x>0$ is {\em large enough}", means there exists a constant $k>0$, not depending on $(r,d,\beta,\alpha,\epsilon)$, such that $A$ is true for any $x\in(k,\infty)$.
\end{definition}

Some additional notation is applied in the following proofs:
\begin{itemize}
\item$M_\phi := 3d+{1\over \beta}$. We have $M_\phi > \max_x \phi(x)$ and $M_\phi>\min \{x:\phi(x)= 0\} +d>1$.
\item For any $\mathcal{X}, \mathcal{X}'\subseteq \mathcal{C}$, \\$\Gamma_{\mathcal{X},\mathcal{X}'}: = \sum_{C\in\mathcal{X}}\sum_{C':C'\in\mathcal{X}'} \Gamma_{C,C'}$, where $\Gamma_{C,C'}$ is  defined in \eqref{eq:defGamma}. 
\item $D_C$ is defined by $$ D_C:=
\begin{cases}
\sum_{i:i\in\mathcal{F}}\Gamma_{C,C\cup \{i\}} & \text{if }C\neq \mathcal{F}\\
\gamma x_{\mathcal{F}} &\text{if }C=\mathcal{F}, \gamma<\infty\\
0& \text{if }C = \mathcal{F},\gamma=\infty
\end{cases}.
$$
Except in the case $C=\mf$ and $\gamma=\infty$, $D_C$ is the aggregate rate that peers leave the group of type $C$ peers.
\item For any $\mathcal{X}\subseteq \mathcal{C}$, $x_{\mathcal{X}} := \sum_{C:C\in\mathcal{X}}x_C$,  $D_{\mathcal{X}} := \sum_{C:C\in\mathcal{X}} D_C$, $D_{total}: =D_{\mathcal{C}}$, $\lambda_{\mathcal{X}} := \sum_{C:C\in\mathcal{X}} \lambda_C$,
 $\lambda_{\mathcal{X}}^{\ast} = \sum_{C:C\in\mathcal{X}}\lambda_{C}(K-|C|+\mu/\gamma)$.
\end{itemize}

Now we start to prove that $W$ given by \eqref{eq:defW} or \eqref{eq:defW2} is a valid Lyapunov function.
The following proof applies if either $0<\mu<\gamma<\infty$ or $0<\mu<\gamma=\infty$, with differences being
stated when necessary.

To begin, we identify a simple approximation to the drift of $W$. Notice that $Q(\ast)$ is linear, so if $0<\mu<\gamma<\infty$,
$$
Q(W) = \sum_{C:C\in \mathcal{C}} r^{|C|} Q(T_C),
$$
where
$$
 Q(T_C) =\begin{cases} {1\over2}Q(E_C^2) + \alpha Q(E_C\phi(H_C)) &\text{if }C\neq\mathcal{F} \\
{1\over2}Q(n^2) & \text{if }C=\mathcal{F}
\end{cases}.
$$
If $0<\mu<\gamma=\infty$,
\ben
Q(W) = \sum_{C:C\in \mathcal{C}-\{\mf\}} r^{|C|} Q(T_C).
\een
Define $\LL W$, an approximation of $Q(W)$, as follows: If $0<\mu<\gamma<\infty$,
$$
\LL W:= \sum_{C:C\in \mathcal{C}} r^{|C|} \LL T_C ,
$$
where
$$
 \LL T_C :=\begin{cases} E_CQ(E_C) + \alpha E_CQ(\phi(H_C))&\text{if }C\neq \mathcal{F}\\
nQ(n) &\text{if }C=\mathcal{F}\end{cases} . \label{eq:dotW}
$$
If $0<\mu<\gamma=\infty$,
\be
\LL W:&=& \sum_{C:C\in \mathcal{C}-\{\mf\}} r^{|C|} \LL T_C.  \label{eq;dotW2}
\ee

The following lemma provides a bound on the approximation error:
\begin{lemma} \label{lem:QWLW}
$|Q(W)-\LL W|\leq M_\phi (D_{total}+1) \Theta(1)$. 
\end{lemma}
\begin{proof}
Compare $Q(W)$ and $\LL W$ term by term. Consider terms of the form  $Q(T_C)$ and $\LL T_C$. First assume $C\neq \mf.$
Because $\alpha<1$, we can write
$$
\left|Q(T_C)-\LL T_C\right|\leq a_1+a_2+a_3,
$$
where
\ben
a_1 & = &  \left|{1\over2}Q(E_C^2) - E_CQ(E_C)\right| \\
& = & \lambda_{\mE_C}+\Gamma_{\mE_C,\mH_C}\leq \lambda_{total}+D_{total}\\
 a_2 & = &  \bigg|Q(E_C\phi(H_C)) -\\&&~~~ \bigg[Q(E_C)\phi(H_C)+E_CQ(\phi(H_C))\bigg]\bigg| \\
a_3 &=&  |Q(E_C)\phi(H_C)|\leq M_{\phi}(\lambda_{\mE_C}+\Gamma_{\mE_C,\mH_C}).
\een
The only way  $E_C$ and $\phi(H_C)$ can simultaneously change is that some peer with type in $\mE_C$ becomes a peer with type in $\mH_C$, causing $E_C$ to  decrease by $1$, and $\phi(H_C)$  to decrease by at most ${K+\mu/\gamma\over 1-\mu/\gamma}$, so $$a_2 \leq {K+\mu/\gamma\over 1-\mu/\gamma}\Gamma_{\mE_C,\mH_C}.$$ From the discussion above and the fact that $\Gamma_{\mE_C,\mH_C}\leq D_{total}$, we have 
\be |Q(T_C)-\LL T_C| & \leq  &  M_\phi \Theta(1) + M_\phi D_{total}\Theta(1) \nonumber \\
& =  & M_\phi (D_{total}+1)\Theta(1) \label{eq:eq11}
\ee
for every $C\in\mathcal{C}-\{\mf\}$. 

Second, assume $C=\mf$ and $\gamma<\infty.$   Then,
\ben
\left|Q(T_C) - \LL T_C\right| & = & \left|{1\over2}Q(n^2) -n Q(n)\right|  \\
&= & \lambda_{total}+D_{\mf}\leq \lambda_{total}+D_{total},
\een
which implies \eqref{eq:eq11} for $C=\mf$. There are only finitely many terms of $T_C$ in $W$ ($2^K$ in total), and $r<1$: Lemma \ref{lem:QWLW} follows. \end{proof}

Now we offer Lemma \ref{lem:QQQQ1} and Lemma \ref{lem:LTSt}, both concerning upper bounds of  $\LL T_C$. They are applied for the proof of Lemma \ref{prop:LW}.

\begin{lemma} \label{lem:QQQQ1}
If $d$ is large enough, $Q(E_C)\leq \Theta(1), Q(\phi(H_C))\leq M_\phi\Theta(1)$, $\LL T_C\leq M_\phi\Theta(E_C)\leq M_\phi \Theta(n)$ for any $C\in\mathcal{C}$.
\end{lemma}
\begin{proof}
The upper bound for the drift of $E_C$ is  obvious: $Q(E_C)\leq \lambda_{\mE_C}\leq \lambda_{total}$.
Next consider $Q(\phi(H_C))$. Since $H_\mf\equiv 0$, we restrict our attention to the case $C\neq \mf$. Because $\phi$ is a decreasing function, only the rate for $H_C$ to decrease contributes to the positive part in the drift of $\phi(H_C)$, so to consider an upper bound of $Q(\phi(H_C))$ it satisfies to consider the rates of transitions that decrease $H_C$. There are two ways $H_C$ can decrease:  peers with one type in $\mH_C$ becoming another type in $\mH_C$ -- with aggregate rate $\Gamma_{\mH_C,\mH_C}$, and peer seeds departing -- with rate $D_\mf$. Because the maximum that $\phi(H_C)$ can jump up is less than or equal to ${1+\mu/\gamma\over 1-\mu/\gamma}$, an upper bound for the drift of $\phi(H_C)$ is
\ben
Q(\phi(H_C)) &\leq &  {1+\mu/\gamma\over 1-\mu/\gamma}(\Gamma_{\mH_C,\mH_C}+D_\mf) \\
&\leq & {1+\mu/\gamma\over 1-\mu/\gamma}\left[U_s+H_C\left(\mu+\gamma\mathbf{1}_{\{\gamma<\infty\}}\right)\right] \\
&=    &  \Theta(1)+H_C\Theta(1).
\een

We can choose $d$ large enough, i.e. $d>{1+\mu/\gamma\over 1-\mu/\gamma}$, so $M_\phi>2d+1/\beta + {1+\mu/\gamma\over 1-\mu/\gamma}$. Thus $Q(\phi(H_C))$ vanishes when $H_C>M_\phi$, because $\phi(H_C)$ vanishes when $H_C>2d+1/\beta$ and the jump size of $H_C$ is bounded below by $-{1+\mu/\gamma\over 1-\mu/\gamma} \geq -d$. Hence 
$$Q(\phi(H_C))\leq \Theta(1)+M_\phi\Theta(1)\in M_\phi\Theta(1),$$
because $M_\phi>1$.

Finally, the bound on $\LL T_C$ follows from the other two bounds already proved. Hence, Lemma \ref{lem:QQQQ1} is proved. \end{proof}

\begin{lemma} \label{lem:LTSt}
If $d$ is large enough, $1-\alpha,\epsilon M_\phi, \beta$ are small enough and $\beta\left({K+\mu/\gamma\over 1-\mu/\gamma}\right)^2\leq \frac{1}{\alpha}-1$, for any $S\in\mathcal{C}-\{\mathcal{F}\}$ and any nonzero state $\mathbf{x}$ such that $x_S/n>1-\epsilon$,
\be
\LL T_S\leq {1\over2}\triangle_S E_S. \label{eq:toshow}
\ee
\end{lemma}
\begin{remark}  \label{remark_Lyapunov_function}
Recall that $\LL T_S = E_S[Q(E_S)+\alpha Q(\phi(H_S))],$ where the term $E_SQ(E_S)$ can be traced back to
the quadratic term $\frac{1}{2}E_S^2$ of $W,$ and $\alpha E_S Q(\phi(H_S))$ can be traced back to the
term $\alpha E_SQ(\phi(H_S))$ of $W.$
Before giving the proof of Lemma \ref{lem:LTSt}, we describe why the term $\alpha E_SQ(\phi(H_S))$ is needed
and how it helps $\LL T_S$ be negative. It has been discussed that the worst distribution of heavy load is when the heavy load aggregates in a type with only one missing piece. Consider the case $|S|=K-1$. Notice that $E_SQ(E_S) = E_S(\lambda_{\mE_S}-\Gamma_{\mE_S,\mH_S})$ and $\Gamma_{\mE_S,\mH_S}\geq D_S \geq {x_S\over n}[U_s+H_S\mu{1-\mu/\gamma\over K+\mu/\gamma}]$. Here we assume ${x_S\over n}\geq 1-\epsilon$. So  $\Gamma_{\mE_S,\mH_S}$ increases almost proportionally to  $H_S$. When $H_S$ is larger than $d$ for $d$ sufficiently large, $\Gamma_{\mE_S,\mH_S}$ is larger than $\lambda_{\mE_S}$, so $E_SQ(E_S)$ is negative and is bounded above by $-\Theta(E_S)=-\Theta(n)$. 
But when $H_S$ is smaller than $d$,  $\Gamma_{\mE_S,\mH_S}$ can be smaller than $\lambda_{\mE_S}$, so $E_SQ(E_S)$ is positive and is lower bounded by $\Theta(E_S)=\Theta(n)$, which has the wrong sign. The term $\alpha E_SQ(\phi(H_S))$ is chosen so that $\alpha Q(\phi(H_S))$  balances out the coefficient $\lambda_{\mE_S}-\Gamma_{\mE_S,\mH_S}$ when $H_S$ is small, so that $\LL T_S$ is still  negative and upper bounded by $-\Theta(E_S)$. 

The definition of $H_S$ implies that, when $x_S$ is close to $n$, $H_S$ is the mean number of type $S$ peers that will be helped by the helping peers, which are the ones in $\mH_S.$ (By saying a peer is helped, we mean a piece is uploaded to the peer). In other words, $H_S$ is the stored potential for helping type $S$ peers.  As type $S$ peers are helped by the helping peers, the potential decreases, with the magnitude of decrease equal to the number of type $S$ peers which are helped. So if we only consider the piece transmissions involving one peer of type $S$ and one peer of type in $\mH_S$, the downward drift of $H_S$ has magnitude less than or equal to the downward drift of $E_S$. If we only consider the external arrivals and the uploads from the fixed seed, the terms in the drift of $H_S$ are ${1\over1-\mu/\gamma}\left[\sum_{C:C\in \mH_S}(K-|C|+\mu/\gamma)\lambda_C +U_s\mu/\gamma\right]$, and the  terms in the drift of $E_S$ are $\lambda_{\mE_S}-U_s$; the former is larger than the latter precisely
because of \eqref{eq:repcd}. Finally, $H_S$ has a bit more downward drift due to peers other than type $S$ peers uploading to
peers in $\mH_S$, but that is small for $\epsilon$ sufficiently small. Combining the downward and the other drifts, we  see that the drift of $H_S$ is approximately the same as the drift of $E_S$, with the drift of $H_S$  a little greater. The difference of the two drifts is $\triangle_S$, defined in \eqref{eq:repcd}.
 Also, when $H_S$ is small, the function $\phi$ at $H_S$ has derivative $-1$. Thus the coefficient of $E_S$ in $\LL T_S$, which is $Q(E_S)+\alpha Q(\phi(H_S))$, is negative because $\alpha$ is close to $1$, so  $\LL T_S$ is upper bounded by $-\Theta(E_S)=-\Theta(n)$.
 
 In summary, the above  explains the reason we included the term $E_S\phi(H_S)$ in the Lyapunov function; it balances out the positive drift of ${1\over2}E_S^2$ when $H_S$ is small. 
\end{remark}
\begin{proof}
Now the detailed proof of Lemma \ref{lem:LTSt} is given. Consider a nonzero state $\mathbf{x}$ of type \Rmnum{1}, with $ S\in\mathcal{C}-\{\mf\}$, $x_S/n>1-\epsilon$. Recall that \eqref{eq:repcd} is assumed to hold; $\triangle_S<0$.  We begin
with three observations.

First, consider a lower  bound for $Q(H_S)$:
\begin{eqnarray}
Q(H_S)&=& {1\over 1-\mu/\gamma}\sum_{C':C'\in\mH_S}(K-|C'|+\mu/\gamma)Q(x_{C'}) \nonumber\\
&\geq & {1\over 1-\mu/\gamma}\left[\lambda_{\mH_S}^\ast+D_S{\mu\over\gamma} - \Gamma_{\mH_S,\mH_S} - D_{\mf}{\mu\over\gamma}\right] \nonumber \\
&=& {1\over 1-\mu/\gamma}\left(\lambda_{\mH_S}^\ast + b_1\right) -D_S, \label{eq:qhs}
\end{eqnarray}
where $b_1: = D_S - \Gamma_{\mH_S,\mH_S} - x_{\mf}\mu$. In view of
\be
D_S & \geq & (1-\epsilon)(U_s +x_{\mH_S}\mu ) \nonumber  \\
& \geq &  U_s+x_{\mH_S}\mu -\epsilon[\Theta(1)+x_{\mH_S}\Theta(1)]  \label{eq:DSg}
\ee
and
\be
\Gamma_{\mH_S,\mH_S} + x_{\mf}\mu & \leq &  \epsilon U_s+ x_{\mH_S}\mu , \label{eq:GHH}
\ee
it follows that
\be
b_1 \geq U_s-\epsilon[\Theta(1)+x_{\mH_S}\Theta(1)].  \label{eq:dds}
\end{eqnarray}
Combining \eqref{eq:qhs} with \eqref{eq:dds}, yields:
\begin{eqnarray}
Q(H_S)\geq  - h_1 -\epsilon[\Theta(1)+x_{\mH_S}\Theta(1)], \label{eq:QHS}\\
h_1:= D_S -{1\over 1-\mu/\gamma}\left(\lambda_{\mH_S}^\ast + U_s\right).  \label{eq:defh1}
\end{eqnarray}

Second,
\be
D_S\geq x_{\mH_S}\mu (1-\epsilon) = x_{\mH_S}\Theta(1) = H_S\Theta(1), \label{eq:D_S}
\ee
because $x_{\mH_S}\leq H_S\leq {K+\mu/\gamma\over 1-\mu/\gamma}x_{\mH_S}$ and $\epsilon<{1\over2}$.

Third, substituting \eqref{eq:D_S} into \eqref{eq:defh1} yields that if $d$ is sufficiently
large, then  $h_1 \geq d\Theta(1) - \Theta(1)$ whenever $H_S >d.$  Therefore, if $d$ is sufficiently large,
\be
h_1 > 0 ~~~\mbox{whenever}~~~ H_S > d.      \label{eq;allh1}
\ee

The remainder of the proof is divided into two, according to the value of $H_S.$
\begin{itemize}
\item {\bf $\bf H_S \leq M_\phi$}: Under this condition, $x_{\mH_S}\leq M_\phi$ and $M_\phi>1$,  so \eqref{eq:QHS}  implies:
\begin{eqnarray}
Q(H_S)\geq -h_1 -\epsilon M_\phi \Theta(1).  \label{eq:hs2}
\end{eqnarray}

Because $\phi'$ exists and is Lipschitz continuous with Lipschitz constant $\beta,$
and because the magnitudes of the jumps of $H_S$ are bounded by
${K+\mu/\gamma\over 1-\mu/\gamma}$, Lemma \ref{lem:expansion} yields
\be
Q(\phi(H_S)) \leq \phi'(H_S)Q(H_S) + b_2 \label{eq:phiHS}
\ee
where
\be
b_2 := {\beta\over2}\left({K+\mu/\gamma\over 1-\mu/\gamma}\right)^2\times~~~~~~~~~~~~~~~~~~~~\nonumber \\
 (\lambda_{\mH_S}+\Gamma_{\mE_S,\mH_S}+\Gamma_{\mH_S,\mH_S}+x_{\mf}\mu).  \label{eq:defb}
\ee
Upper bounds for the terms in the right hand side of  \eqref{eq:phiHS} are found next.
First, a bound for $b_2$ is found. By \eqref{eq:DSg} and \eqref{eq:GHH},
\begin{eqnarray}
\Gamma_{\mE_S,\mH_S} & \leq &  D_S + \epsilon(U_s+x_{\mH_S}\mu)  \nonumber \\
& \leq  & D_S+\epsilon M_\phi \Theta(1); \label{eq:gg}  \\
\Gamma_{\mH_S,\mH_S}+x_{\mf}\mu & \leq &  U_s+x_{\mH_S}\mu  \nonumber  \\
& \leq &  D_S+\epsilon M_\phi \Theta(1). \label{eq:g}
\end{eqnarray}
Substituting \eqref{eq:gg} and \eqref{eq:g} into the right side of \eqref{eq:defb}, yields
\begin{eqnarray}
b_2 & \leq & \beta\Theta(1) +\beta\epsilon M_\phi \Theta(1) +\nonumber \\&&\beta\left({K+\frac{\mu}{\gamma}\over 1-\frac{\mu}{\gamma}}\right)^2 D_S \label{eq:b21}\\
&\leq & \left(\frac{1}{\alpha}-1\right)D_S +\beta\Theta(1), \label{eq:b22}
\end{eqnarray}
where to obtain \eqref{eq:b22} from \eqref{eq:b21}, we assume  $\beta\left({K+\mu/\gamma \over 1-\mu/\gamma}\right)^2 \leq \frac{1}{\alpha}-1$ and $\epsilon M_\phi <1$.

Second, a bound for $\phi'(H_S)Q(H_S)$ is found.  Taking into account that $-1 \leq \phi' \leq 0,$   multiply both sides of  \eqref{eq:hs2} by $\phi'(H_S)$ and use the fact $\phi'(H_s)=-1$ for $H_s \leq d,$ and \eqref{eq;allh1}, to  obtain:
\be
\phi'(H_S)Q(H_S)& \leq&  -\phi'(H_S)h_1 + \epsilon M_\phi \Theta(1) \nonumber \\
& \leq  &  h_1  +\epsilon M_\phi \Theta(1). \label{eq:h1}
\ee
Substituting \eqref{eq:defh1}, \eqref{eq:b22} and \eqref{eq:h1}  into \eqref{eq:phiHS} yields
\begin{eqnarray}
Q(\phi(H_S) ) & \leq & \frac{1}{\alpha} D_s  - {1  \over 1-\mu/\gamma}\left(\lambda_{\mH_S}^\ast + U_s\right) \nonumber \\
&& +\epsilon M_\phi \Theta(1) +\beta\Theta(1) .    \label{eqQES_bound} 
\end{eqnarray}

We obtain a bound on $Q(E_S)+\alpha Q(\phi(H_S))$, the coefficient of $E_S$ in $\LL T_S$, using
 \eqref{eqQES_bound} and the facts $Q(E_S)\leq \lambda_{\mE_S}-D_S$  and $\alpha<1$, as follows:
\begin{eqnarray}
\lefteqn{Q(E_S)+\alpha Q(\phi(H_S)) \nonumber}\\
&\leq& \lambda_{\mE_S} - {\alpha\over 1-\mu/\gamma}\left(\lambda_{\mH_S}^\ast + U_s\right) \nonumber \\
&& ~~~~ +\epsilon M_\phi \Theta(1) +\beta\Theta(1)\nonumber\\
&\leq&\triangle_S +(1-\alpha)\Theta(1) \nonumber \\
&& ~~~~+ \epsilon M_\phi \Theta(1) +\beta\Theta(1). \label{eq:qaq} 
\end{eqnarray}

Because $\triangle_S<0$, if $1-\alpha, \epsilon M_\phi, \beta$ are close to $0$, the last three terms in \eqref{eq:qaq} can be made
small compared to $|\triangle_S|$,  so $Q(E_S)+\alpha Q(\phi(H_S))\leq {1\over2}\triangle_S$,
which implies \eqref{eq:toshow}.

\item {\bf $\bf H_S> M_\phi$}:  To take care of this case, assume  $d>{K+\mu/\gamma\over 1-\mu/\gamma}$, so
$M_\phi>2d+1/\beta + {K+\mu/\gamma\over 1-\mu/\gamma}$. Hence $Q(\phi(H_S))$ vanishes for $H_S$ in this range.
By \eqref{eq:D_S},
\ben
Q(E_S)+\alpha Q(\phi(H_S)) & \leq &  \lambda_{\mE_S} - D_S\\
& \leq&  \Theta(1) - M_\phi \Theta(1) \\&<& {1\over2}\triangle_S ,
\een
if $d$ is large enough so that $M_\phi$ is large enough.. Therefore \eqref{eq:toshow} holds.
\end{itemize}
The proof of Lemma \ref{lem:LTSt} is complete.
\end{proof}

Lemmas \ref{lem:QQQQ1} and \ref{lem:LTSt} will be used to prove the following lemma.
\begin{lemma} \label{prop:LW}
If $d$ is large enough, $(1-\alpha),\beta, rM_\phi, \epsilon M_\phi r^{-K}$ are small enough, and $\beta\left({K+\mu/\gamma\over 1-\mu/\gamma}\right)^2\leq \frac{1}{\alpha}-1$,
\begin{itemize} 
\item[(a)] On class \Rmnum{1}, $\LL W\leq -r^{K} \Theta(n)$;
\item[(b)] On class \Rmnum{2}, $\LL W\leq -r^K \epsilon^3\Theta(n^2)+M_\phi \Theta(n)$.
\end{itemize}
\end{lemma}
\begin{proof} First consider Lemma \ref{prop:LW}(a). Since there are only finitely many types, we can fix a set $S\in \mathcal{C}-\{\mf\}$ and consider the set of class \Rmnum{1} states $\mathbf{x}$ for which $x_S/n>1-\epsilon$. Since $\epsilon\in (0,{1\over2})$, $E_S>{1\over2}n$. By assumption in this section, $\triangle_S<0.$   By Lemma \ref{lem:LTSt},
\be
\LL T_S \leq {1\over4}\triangle_S n\in -\Theta(n). \label{eq:TS4}
\ee
For type $C$ with $|C|> |S|$,  Lemma \ref{lem:QQQQ1} and \eqref{eq:TS4} imply
\be
r^{|C|}\LL T_C  \leq rM_\phi r^{|S|} \Theta(n) < 2^{-K-1}r^{|S|} \big|\LL T_S\big|.  \label{eq:r1}
\ee
if $r M_\phi$ is chosen to be small enough.

For type $C$ with $|C|\leq |S|$ but $C\neq S$, $E_C\leq \epsilon n$; Lemma \ref{lem:QQQQ1} and \eqref{eq:TS4} imply
\be
r^{|C|}\LL T_C &\leq&  r^{|C|}M_\phi \Theta(E_C)\leq \epsilon M_\phi r^{-K} r^{|S|}  \Theta(n) \nonumber \\
& < &  2^{-K-1}r^{|S|} \big|\LL T_S\big| \label{eq:r2}
\ee
if $\epsilon M_\phi r^{-K}$ is chosen to be small enough. 

Equations \eqref{eq:r1} and \eqref{eq:r2} imply that
\ben
\LL W &=& r^{|S|}\LL T_S + \sum_{C: |C|>|S|}r^{|C|}\LL T_C+\\
&&~~~~~\sum_{C:|C|\leq |S|, C\neq S} r^{|C|}\LL T_C \nonumber\\
&\leq& r^{|S|}\LL T_S + {1\over2}r^{|S|}\big|\LL T_S\big| \\
&\leq & {1\over8}r^{|S|}\triangle_S n \leq -r^K \Theta(n),  
\een
which proves Lemma \ref{prop:LW}(a).

Next consider Lemma \ref{prop:LW}(b). First, suppose $C_1\not\subseteq C_2$ and consider the set of class \Rmnum{2} states $\mathbf{x}$ such that $x_{C_1}/n>\eta, x_{C_2}/n>\eta$, where $\eta=\epsilon/2^K$. For such states:
\be
\Gamma_{\mE_{C_2},\mH_{C_2}} \geq D_{C_2}\geq {x_{C_2}\over n} x_{C_1}\mu \geq \mu\eta^2 n \in\epsilon^2\Theta(n). \label{eq:GC_2}
\ee
Since $E_{C_2}\geq x_{C_2}\geq \eta n$, \eqref{eq:GC_2} implies
\be
E_{C_2}Q(E_{C_2}) & = & E_{C_2}(\lambda_{\mE_{C_2}}-\Gamma_{\mE_{C_2},\mH_{C_2}}) \nonumber \\
& \leq &  -\epsilon^3 \Theta(n^2) + \Theta(n). \label{eq:EQC2}
\ee
Lemma \ref{lem:QQQQ1} indicates that $E_{C_2}Q(\phi(H_{C_2}))\leq M_\phi\Theta(n)$, so \eqref{eq:EQC2} implies
\be
\LL T_{C_2} & = &  E_{C_2}Q(E_{C_2}) +\alpha E_{C_2}Q(\phi(H_{C_2})) \nonumber \\
& \leq &  -\epsilon^3\Theta(n^2) + M_\phi \Theta(n). \label{eq:TC2}
\ee
Second, consider the set of class \Rmnum{2} states $\mathbf{x}$ such that $x_{\mathcal{F}}/n>\eta$, where $\eta = \epsilon/2^K$. If $\gamma=\infty$, this set is empty, so suppose $\gamma<\infty$. For such states,
\be
 \LL T_{\mathcal{F}} &= & nQ(n)=n(\lambda_{total}-\gamma x_{\mathcal{F}} )  \nonumber  \\
& \leq & n(\lambda_{total}- \eta \gamma n) \nonumber \\
& = &   -\epsilon\Theta(n^2) +\Theta(n). \label{eq:TF}
\ee

Recall that Lemma \ref{lem:QQQQ1} implies for any $C$, $\LL T_C\leq M_\phi \Theta(n)$.
Therefore, for either condition $C_1\not\subseteq C_2$ or $C_1=C_2=\mathcal{F}$,  \eqref{eq:TC2} and \eqref{eq:TF} imply that, over the set of all class \Rmnum{2} states,
\ben
\LL W & \leq & r^{|C_2|}\LL T_{C_2} +\sum_{C:C\neq C_2}\LL T_C \\
& \leq &  -r^K \epsilon^3\Theta(n^2) + M_\phi \Theta(n),
\een
which proves Lemma \ref{prop:LW}(b). 
\end{proof}

With Lemmas \ref{lem:QWLW} and \ref{prop:LW}, Theorem \ref{thm:main}(b) in the case $0<\mu<\gamma\leq \infty$ can be proved:

\begin{proofof} {\em Theorem \ref{thm:main}(b):} On class \Rmnum{1}, 
\begin{eqnarray*}
D_{total} &\leq & D_S + \sum_{C:C\neq S}D_C\\
& \leq & U_s+x_{\mH_S}\mu + \sum_{C:C\neq S} {x_C\over n}(U_s+n\mu) \\
& \leq  & 2 (U_s+\epsilon n\mu) =\Theta(1)+\epsilon\Theta(n).
\end{eqnarray*}
So  Lemma \ref{lem:QWLW} implies that on class \Rmnum{1}, 
\ben
|Q(W)-\LL W|\leq \epsilon M_\phi \Theta(n) + M_\phi \Theta(1).
\een 
Combining with Lemma \ref{prop:LW}(a), implies that under the conditions of Lemma \ref{prop:LW}, on class \Rmnum{1},
\be Q(W)&\leq& \LL W+ |Q(W)-\LL W| \nonumber \\&  \leq  &  -r^{K} \Theta(n)+ \epsilon M_\phi \Theta(n) + M_\phi \Theta(1) \nonumber \\
&\in& -r^{K} \Theta(n)+ M_\phi \Theta(1).  \label{eq:w1}
\ee
if $\epsilon M_\phi r^{-K} $ is small enough. 

On class \Rmnum{2}, $D_{total}\leq U_s+n\mu=\Theta(n)$, so  Lemma \ref{lem:QWLW} implies that 
\ben
|Q(W)-\LL W|\leq M_\phi \Theta(n).
\een
Combining with Lemma \ref{prop:LW}(b), implies that under the conditions of Lemma \ref{prop:LW}, on class \Rmnum{2},
\be
Q(W) & \leq &  \LL W+ |Q(W)-\LL W| \nonumber \\
& \leq  &  -r^K \epsilon^3\Theta(n^2) + M_\phi \Theta(n).  \label{eq:w2}
\ee

Equations \eqref{eq:w1} and \eqref{eq:w2} imply that if $(r,d,\beta,\alpha, \epsilon)$ satisfies the conditions of Lemma \ref{prop:LW}, there exists $\xi>0$ sufficiently small such that $Q(W)\leq-\xi n$ for all  $n$ larger than some constant. For such $\xi$ and such $(r,d,\beta,\alpha)$, $W$ is a valid Lyapunov function, so by Lemma \ref{prop:FL}, Theorem \ref{thm:main}(b) for the case $0<\mu<\gamma\leq\infty$ is proved.
\end{proofof}

\subsection{Proof of Positive Recurrence when $0<\gamma\leq\mu$ in Theorem \ref{thm:main} }

Now we consider the case when $0<\gamma\leq \mu$. 
Assume $U_s+\sum_{C:k\in C}\lambda_C>0$ for all $k\in\mf$. 
Then  $U_s+\lambda^\ast_{\mH_C}>0$ for any $C\in\mathcal{C}-\{\mf\}$. 
Consider a Lyapunov function of the following form:
\begin{equation}
W': = \sum_{C:C\in\mathcal{C}} r^{|C|}T_C',  \label{eq:WC}
\end{equation}
where
$$T_C': = \begin{cases} {1\over2}E_C^2+pE_C\phi(H_C')& \text{if } C\neq \mf\\
{1\over2}n^2 & \text{if } C=\mf
\end{cases}, 
$$
$H_C':=\sum_{C':C'\in\mH_C}(K+1-|C'|)x_{C'},$ and
$p$ is a constant (i.e. $p=\Theta(1)$) such that
\be
\lambda_{\mE_C}-p(U_s+\lambda^\ast_{\mH_C}) <0, \forall  C\in\mathcal{C}-\{\mf\}. \label{eq:p}
\ee
The variable $\alpha$ is not used in this section, so the big $\Theta$ notation is uniform in $(r, \beta, d, ,\epsilon).$

Define $\LL W'$, as follows:
\begin{equation}
\LL W':= \sum_{C:C\in \mathcal{C}} r^{|C|} \LL T_C' , \label{eq:LWC}
\end{equation}
where
$$  \LL T_C' :=\begin{cases} E_CQ(E_C) + p E_CQ(\phi(H_C'))&\text{if }C\neq \mathcal{F}\\
nQ(n) &\text{if }C=\mathcal{F}\end{cases}. 
$$

Lemmas \ref{lem:QWLW} and \ref{lem:QQQQ1} can be verified as before,
with $H_C,$  $W,$  $\LL W,$ and $\LL T_C$ replaced by $H_C',$  $W',$  $\LL W',$ and $\LL T_C'$, respectively.
The following lemma similar to Lemma \ref{lem:LTSt} can be established:
\begin{lemma} \label{lem:LTS1}
If $d$ is large enough, $\epsilon M_\phi, \beta$ are small enough, for any $S\in\mathcal{C}-\{\mathcal{F}\}$ such that $x_S/n>1-\epsilon$,
\be
\LL T_S'\leq {1\over2}[\lambda_{\mE_S}-p(U_s+\lambda^\ast_{\mH_S})] E_S. \label{eq:toshow1}
\ee
\end{lemma}
\begin{proof}
Suppose $S\in\mathcal{C}-\{\mf\}$ and $x_S/n>1-\epsilon$, $\epsilon\in(0,{1\over2})$,  one lower bound for $Q(H_S')$ is:
\be
Q(H_S') &\geq& \lambda_{\mH_S}^\ast + D_S -\Gamma_{\mH_S,\mH_S}-x_{\mf}\gamma  \\
& \geq &  \lambda^\ast_{\mH_S} + b_1, \label{eq:HS'}
\ee
where $b_1$ was introduced in \eqref{eq:qhs},  because $\gamma\leq \mu$. Substituting \eqref{eq:dds} into \eqref{eq:HS'}
yields
\be
Q(H_S') \geq \lambda_{\mH_S}^\ast +U_s -\epsilon[\Theta(1)+x_{\mH_S}\Theta(1)].  \label{eq:HS11}
\ee

Consider two conditions of $H_S'$:
\begin{itemize}
\item {\bf $\bf H_S' \leq M_\phi$}: Under this condition, $x_{\mH_S}\leq M_\phi$ and $M_\phi>1$,  so \eqref{eq:HS11}  becomes:
\begin{eqnarray}
Q(H_S')\geq U_s+\lambda_{\mH_S}^\ast -\epsilon M_\phi \Theta(1).  \label{eq:hs22}
\end{eqnarray}

Because $\phi'$ exists and is Lipschitz continuous with Lipschitz constant $\beta,$ 
and because the magnitude of the jump of $H_S'$ is bounded by $K+1$, by Lemma \ref{lem:expansion},
\begin{equation}
Q(\phi(H_S')) \leq \phi'(H_S')Q(H_S') + b_2'  \label{eq:QHSS}
\end{equation}
where
$$
 b_2' := {\beta\over2}\left(K+1\right)^2(\lambda_{\mH_S}+\Gamma_{\mE_S,\mH_S}+\Gamma_{\mH_S,\mH_S}+x_{\mf}\mu). \nonumber
$$
Consider the term $b_2'$. By  \eqref{eq:gg} and \eqref{eq:g}, and assuming $\epsilon M_\phi <1$, we have
\begin{eqnarray}
b_2'&\leq& \beta\Theta(1) +\beta\epsilon M_\phi \Theta(1) + \beta D_S\Theta(1)\nonumber \\
&\leq&\beta D_S \Theta(1)+\beta\Theta(1). \label{eq:b2d}
\end{eqnarray}
If $\beta$ is small enough, $\beta D_S\Theta(1)<{1\over2p}D_S$, so \eqref{eq:b2d} becomes:
\be
b_2'\leq {1\over2p}D_S +\beta\Theta(1).\label{eq:b2dd}
\ee

Substituting \eqref{eq:hs22} and \eqref{eq:b2dd}  into \eqref{eq:QHSS}, and applying $Q(E_S)\leq \lambda_{\mE_S}-D_S$, we can  bound $Q(E_S)+p Q(\phi(H_S'))$, the coefficient of $E_S$ in $\LL T_S'$, as follows:
\begin{eqnarray}
\lefteqn{Q(E_S)+p Q(\phi(H_S')) }\nonumber  \\
&\leq& \lambda_{\mE_S} - {1\over2}D_S + p \phi'(H_S')(U_s+\lambda_{\mH_S}^\ast) + \nonumber \\
&&~~~\beta \Theta(1) +\epsilon M_\phi \Theta(1) \nonumber\\
&=& \lambda_{\mE_S} - p(U_s+\lambda_{\mH_S}^\ast) +b_3'+  \nonumber  \\
&&~~~\beta \Theta(1) +  \epsilon M_\phi \Theta(1), \label{eq:qaq1}
\end{eqnarray}
where 
\begin{eqnarray*}
b_3'&: =& p(1+\phi'(H_S'))(U_s+\lambda_{\mH_S}^\ast) - {1\over2}D_S \\
&\leq&
\begin{cases}
-{1\over2} D_S & \text{if } H_S'<d\\
\Theta(1) - d\Theta(1) & \text{if }H_S'\geq d
\end{cases}.
\end{eqnarray*}
because $D_S\geq (1-\epsilon)x_{\mH_S}\mu\geq {1\over2 (K+1)}H_S' \mu =\Theta(H_S')$. Hence if $d$ is large enough, $b_3'\leq 0$. If $\beta,\epsilon M_\phi$ are close to $0$, the last two terms in \eqref{eq:qaq1} can be neglected. Thus, $Q(E_S)+p Q(\phi(H_S'))\leq {1\over2}[\lambda_{\mE_S} - p(U_s+\lambda_{\mH_S}^\ast)]$, which implies \eqref{eq:toshow1}.

\item {\bf $\bf H_S'> M_\phi$}: Under this condition, choose $d$ such that $d>K+1$, so $M_\phi>2d+1/\beta + K+1$. Hence $Q(\phi(H_S'))$ vanishes for $H_S'$ in this range. The fact that $D_S=\Theta(H_S')$ yields that
\begin{eqnarray*}
\lefteqn{Q(E_S)+p Q(\phi(H_S')) }\\
& \leq & \lambda_{\mE_S} - D_S\\
& \leq &  \Theta(1) - M_\phi \Theta(1) \\
&< &{1\over2}[\lambda_{\mE_S} - p(U_s+\lambda_{\mH_S}^\ast)] ,
\end{eqnarray*}
if $d$ is large enough, and hence also $M_\phi$. Therefore \eqref{eq:toshow1} holds.
\end{itemize}
So far, Lemma \ref{lem:LTS1} is proved.
\end{proof}

With Lemma \ref{lem:QWLW}, \ref{lem:QQQQ1} and \ref{lem:LTS1}, Lemma \ref{prop:LW} with $\LL W$ replaced by $\LL W'$ can be easily verified to be valid. Thereby Theorem \ref{thm:main}(b) at condition $0<\gamma\leq\mu$ is proved.

\section{Extensions}  \label{sec:extensions}

\subsection{General Piece Selection Policies}
A piece selection policy is used to choose which piece is transfered whenever one peer or the fixed seed
is to upload a piece to a chosen peer.
The random useful piece selection policy is assumed in Theorem \ref{thm:main}, but the theorem can be extended
to a large class of piece selection policies.  Such extension was noted in \cite{HajekZhu10_full} for the less general model of that paper.   Essentially the only restriction needed is that if the uploading peer or fixed seed has a useful piece for the
downloading peer, then a useful piece must be transferred.     This restriction is similar to a work
conserving restriction in the theory of service systems.     In particular, Theorem 1 extends to cover a broad class of
rarest first piece selection policies.     Peers can estimate which pieces are more rare in a distributed way,
by exchanging information with the peers they contact.   Even more general policies would allow the piece
selection to depend in an arbitrary way on the piece collections of all peers.

To be specific, consider the following family $\cal H$ of piece selection policies.   Each policy
 in $\cal H$
corresponds to a mapping $h$ from ${\cal C} \times ({\cal C}  \cup \{ {\cal F} \}  ) \times {\cal S}$
to the set of probability distributions on $\cal F,$   satisfying the usefulness constraint:
$$ \sum_{i\in B-A}  h_i(A,B,\mathbf{x})=1~~\mbox{whenever}~ B \not\subset A, $$
with the following meaning of $h$:
\begin{itemize}
\item When a type $A$ peer is to download a  piece from a type $B$ peer and the state of the entire network is $\mathbf{x},$
piece $i$ is selected with probability $h_i(A,B,\mathbf{x}),$  for $i \in {\cal F}.$
\item When a type $A$ peer is to download a  piece from the fixed seed and the state of the entire network is $\mathbf{x},$
piece $i$ is selected with probability  $h_i(A,{\cal F},\mathbf{x}),$  for $i \in {\cal F}.$
\end{itemize}

Theorem  \ref{thm:main} can be extended to piece selection policies in $\cal H$.  One minor
change is needed, because the Markov process may not be irreducible for some piece selection
policies.     In general, the set of all states that are reachable from the empty
state is the unique minimal closed set of states, and the process restricted to that set of states
is irreducible.     For example, if the lowest numbered useful piece is selected at each download opportunity,
then the minimal closed set of states consists of the states such that each peer has either no pieces
or a consecutively numbered set of pieces beginning with the first piece.
See  \cite{HajekZhu10_full} for further discussion.   We state the result as a theorem.
\begin{theorem}   (Stability conditions for general useful piece selection policies)
Consider the network model of Section \ref{sec:Model}, except with the random piece
selection policy replaced by a policy $h$ in $\cal H.$  
(a)  If either of the two conditions in Theorem \ref{thm:main}(a) hold then the Markov process
 is transient, and the number of peers in the system converges to infinity with probability one.
(ii)      If either of the two conditions in Theorem \ref{thm:main}(b)  hold  the Markov process
restricted to the closed set of states is positive recurrent,
the mean time to reach the empty state from any initial state has finite mean, and the equilibrium distribution $\pi$ is
such that $\sum_{ \mathbf{x}}  \pi(\mathbf{x})  |\mathbf{x}| < \infty.$   
\end{theorem}
Thus, with the possible exception of the borderline case,
 rarest first piece selection does not increase the
region of stability.

\subsection{Network Coding}

Network coding, introduced by Ahlswede, Cai, and Yeung,  \cite{AhlswedeCaiLiYeung},
can be naturally incorporated into P2P distribution networks, as noted in
\cite{GkantsidisRodriguez05}.    The related work \cite{DebMedardChoute06} considers
all to all exchange of pieces among a fixed population of peers through random contacts
and network coding.    The method
can be described as follows.    The file to be transmitted is divided into
$K$ data pieces, $m_1, m_2, \ldots  , m_K.$     The data pieces are
taken to be vectors of some fixed length $r$ over a finite field $\mathbb{F}_q$ with
$q$ elements, where $q$ is some power of a prime number.
If the piece size is $M$ bits, this can be done by viewing each message as an
$r=\lceil M / \log_2(q) \rceil$ dimensional vector over $\mathbb{F}_q.$
Any coded piece $e$ is a linear combination of
the original $K$ data pieces:
$e=\sum_{i=1}^K \theta_i m_i;$
the vector of coefficients $(\theta_1, \ldots , \theta_K)$ is called
the {\em coding vector} of the coded piece; the coding vector
is included whenever a coded piece is sent.
Suppose the fixed seed uploads coded pieces to peers, and peers
exchange coded pieces.      In this context, the type of a peer $A$ is the
subspace $V_A$ of $\mathbb{F}_q^K$
spanned by the coding vectors of the coded pieces it has
received.   Once the dimension of $V_A$ reaches $K$, peer $A$
can recover the original message. Let $\cal V$ denote the set of all
subspaces of $\mathbb{F}_q^K,$ so $\cal V$ is the set of possible
types.

When peer A contacts peer B, suppose peer B sends peer
A a random linear combination of its coded pieces, where
the coefficients are independent and uniformly distributed
over  $\mathbb{F}_q.$   Equivalently, the coding vector of the
coded piece sent from $B$ is uniformly distributed over
$V_B.$      The coded piece is considered useful to $A$ if adding
it to $A$'s collection of coded pieces increases the dimension of $V_A.$
Equivalently, the piece from $B$ is useful to $A$ if its coding vector is
not in the subspace $V_A \cap V_B.$    The probability the piece is useful to
$A$ is therefore given by
\begin{eqnarray*}
\lefteqn{P\{\mbox{piece from $B$ is useful to $A$}\}} \\& = & 1 -  \frac{|V_A\cap V_B|}{|V_B|} \\
& = &1-q^{  dim(V_A\cap V_B)  - dim(V_B) }.
\end{eqnarray*}
If peer $B$ can possibly help peer $A$, meaning $V_B\not\subset V_A$
(true, for example, if $dim(V_B) > dim(V_A)$),
the probability that a random coded piece from B is helpful to A is greater than or equal to $1-\frac{1}{q}.$
Similarly, the probability a random coded piece from the seed is
useful to any peer $A$ with $dim(V_A)\leq K-1$
is also greater than or equal to $1-\frac{1}{q}.$

The network state $\mathbf{x}$ specifies the number of peers in the network of each type. 
There are only finitely many types, so the overall state space is still countably infinite.  
Moreover, the Markov process is easily seen to be irreducible.  A proof of the following variation
of Theorem \ref{thm:main} is summarized below.  Let $\widetilde{\mu}=\left( 1-\frac{1}{q}\right) \mu.$

\begin{theorem}  \label{thm_netcode}   (Stability conditions for a network coding based system)
Suppose random linear network coding with vectors over $\mathbb{F}_q^K$ is used,
with random peer contacts and parameters $K$, $q,$ $(\lambda_V: V\in {\cal V})$, $U_s,$ $\gamma,$ and $\mu.$
Suppose $\lambda_{\mathbb{F}_q^K}=0$ if $\gamma=\infty$, and $\lambda_{total}>0.$  \\
(a) The Markov process is transient if either of the following two conditions is true:
\begin{itemize}
\item $0<\mu<\gamma\leq\infty$ and for some $V^-\in {\cal V}$ with $dim(V^-)=K-1,$
$$
\lambda_{total}> \frac{ U_s+\sum_{V\not\subset V^-} \lambda_V(K-dim(V)+1) }{1-\frac{\mu}{\gamma}};  
$$
\item $0<\gamma\leq\mu,$ $U_s=0,$  and  $\{V\in {\cal V}: \lambda_V > 0 \}$ does not span $\mathbb{F}_q^K.$
\end{itemize}
(b) The process is positive recurrent and $E[n]<\infty$ in equilibrium, if either of the following two conditions is true:
\begin{itemize}
\item $0< \widetilde{\mu} <\gamma \leq \infty$ and for any $V^-\in {\cal V}$ with $dim(V^-)=K-1,$
\begin{eqnarray}
\lefteqn{ \lambda_{total} <  \nonumber } \\
&& \left[ U_s+\sum_{V:  V\not\subset V^-}
 \lambda_V \left( K-dim(V)  + \frac{q}{q-1}  \right)  \right] \nonumber \\
 && \times \left(   
 { 1-\frac{1}{q} \over 1-\frac{\widetilde{\mu} }{\gamma}} \right);  \label{eq:pcd_netcode}
\end{eqnarray}
\item $0<\gamma\leq \widetilde{\mu} $ and either $U_s >0 $ or $\{V\in {\cal V}: \lambda_V > 0 \}$ spans  $\mathbb{F}_q^K.$
\end{itemize}
\end{theorem}

The gap between the necessary and sufficient conditions in Theorem \ref{thm_netcode} can be made arbitrarily small by
taking $q$ large enough.  

For the case that peers arrive with pieces, network coding is quite effective at reducing the
impact of the missing piece syndrome.   For example, suppose peers with no pieces arrive at
rate $\lambda_0$ and peers with one piece arrive at rate $\lambda_1$, where the coding vector
for the piece given to a peer at time of arrival is uniformly distributed over all $q^K$ possibilities.
(So with probability $q^{-K}$ the coding vector is the all zero vector and the piece is useless.)
Suppose there are no other arrivals, that $U_s=0,$ and $\gamma=\infty.$
Thus, the total arrival rate is $\lambda_{total}=\lambda_0+\lambda_1,$ and the fraction of peers arriving
with one (possibly useless) piece is $f=\frac{\lambda_1}{\lambda_0+\lambda_1}.$   Then Theorem \ref{thm_netcode}
yields that the Markov process is transient if $f < \frac{q}{(q-1)K}$ and positive recurrent if
$f >  \frac{q^2}{(q-1)^2K}.$   For example, if $q=64$ and $K=200$, the Markov process is
transient if $f  \leq \frac{1.014}{K} = 0.00507$ and positive recurrent if $f \geq \frac{1.032}{K}=0.00516.$
In contrast, without network coding and a fraction $f$ of peers arriving with one
uniformly randomly selected data piece, Theorem \ref{thm:main} implies the network is transient for any $f  < 1.$

We comment briefly on how the proof of Theorem \ref{thm:main} can be modified to yield
Theorem \ref{thm_netcode}.    First, consider how the proof of Theorem \ref{thm:main}(a) can be modified
to prove Theorem \ref{thm_netcode}(a).  Consider the main case, $0 < \mu < \gamma \leq \infty$.
Let $V^-$ be the subspace of $\mathbb{F}_q^K$ with dimension $K-1$ appearing in part (a).
To incorporate network coding,
the partition of peers described in Section \ref{sec:outline_of_proof} should be replaced by the following
partition:
\begin{itemize}
\item{\em Normal young peer}: A normal young peer is a peer $A$ such that $V_A$ is a proper subset of $V^-.$
\item{\em Infected peer}: An infected peer is a peer $B$ that was a normal young peer when it first arrived,
but at the current time, $V_B\not\subset V^-.$
\item{\em Gifted peer}: A gifted peer is a peer $G$ such that at the time of its arrival, $V_G \not\subset V^-.$
\item{\em One-club peer}: A peer of type $V^-.$
\item{\em Former one-club peer}: A former one-club peer is a peer in the system that is not a one-club peer but at some earlier time was a one-club peer.
\end{itemize}
For any $N_o\geq 1$, it is possible to reach the state with $N_o$ one-club peers and no other
peers in the network.

Call a peer $A$ {\em enlightened} if $V_A \not\subset V^-.$  Note that gifted peers are
enlightened when they arrive, and every other peer must become enlightened before departing.
A peer becoming enlightened with network coding is analogous to a peer
downloading the missing piece without network coding.  In particular,
for the proof of Theorem \ref{thm_netcode}, the process $D_t$ should be the cumulative
number of downloads causing the recipient peers to become enlightened.  

The same autonomous branching system (ABS) can be used as in the proof of Theorem \ref{thm:main}.
Lemma \ref{lemma.onecompare} remains true, but the coupling argument used to prove it becomes
more subtle.  The issue is that the rate that a group (b) or (g) peer downloads pieces can be less than
$\mu(1-\xi),$  because random linear combinations are sent that are not always useful.
This effect causes the group (b) and group (g) peers to remain in the system longer, so that they can
continue to upload useful pieces to one club peers for longer.  However,
note that if $A$ is a group (b) or (g) peer that is not a peer seed, and $B$ is a one-club peer, then the
probability a random piece from $A$ is useful to $B$ is less than or equal to the probability
a random piece from $B$ is useful to $A$.   Therefore,  if the internal clocks of the group (b) and
(g) peers are slowed down so that their download rate of useful pieces matches that of the original
system, then their upload rate of useful pieces to the one club peers will still be at least as large
as in the original system.  

The other parts of the proof of Theorem \ref{thm:main}(a) readily carry over to imply Theorem
\ref{thm_netcode}(a).

Next, the modifications of Theorem \ref{thm:main}(b) needed to prove Theorem \ref{thm_netcode}(b)
are described.   The same approach works with the same form of Lyapunov function, except $\cal V$
is used as the set of types instead of $\cal C.$   In places the cardinality $|C|$ of a type $C$ is used in the proof
of Theorem \ref{thm:main}(b), the dimension $dim(V)$ of a type $V$ is used in the proof of
Theorem \ref{thm_netcode}(b).   In some of the places that $\mu$ is used in the proof
of Theorem \ref{thm:main}(b), it should be replaced by $\widetilde{\mu}. $

The condition \eqref{eq:pcd_netcode} holding for all $V^- \in {\cal V}$ is equivalent to the following:
for any $S\in\mathcal{V}- \mathbb{F}_q^K$, $\triangle_S < 0,$ where
\begin{eqnarray*}
\lefteqn{\triangle_S  =  \sum_{V:V\subseteq S,V\in\mathcal{V}}\lambda_V }  \\
&&  - \left[U_s+  \sum_{V:V\not\subseteq S, V\in\mathcal{V}}\lambda_V\left(K-dim(V)+\frac{\mu}{\gamma}\right)\right] \\
&&~~~~~~~~~\times
 \left(   {1-\frac{1}{q}\over 1-\frac{\widetilde{\mu}}{\gamma}} \right) .
\end{eqnarray*}
The condition  $\triangle_S < 0$ means that the rate of arrival of peers that can become type $S$ peers
is less than a lower bound on the long term rate that peers of type $S$ receive useful pieces.
The particular Lyapunov function we  use in case $0<\widetilde{\mu} <\gamma<\infty$, is:
\begin{equation}
W : = \sum_{V:V\in \mathcal{V}} r^{dim(V)} T_V,     \label{eq:defW}
\end{equation}
where
$$T_V :=
\begin{cases}
{1\over2}E_V^2 + \alpha E_V\phi(H_V)& \text{ if }V\neq \mathbb{F}_q^K\\
{1\over2}n^2&\text{ if }V=\mathbb{F}_q^K
\end{cases}.
$$
with $\alpha, r,d,\beta,$ and $d$ and the function $\phi$ as in the proof of Theorem \ref{thm:main}(b), and
\begin{itemize}
\item$\mE_V: = \{V':V'\subseteq V\}$, which is the collection of types of peers which are or can become type $V$ peers.
\item $\mH_V: = \{V': V'\in\mathcal{V},V'\not\subseteq V\}$,
which is the collection of types of peers which can help type $V$ peers. Notice that
$\mathbb{F}_q^K\in\mH_V$ for any $V\in\mathcal{V}-\mathbb{F}_q^K$ and $\mH_{\mathbb{F}_q^K}= \emptyset$.
\item$E_V: = \sum_{V':V'\in\mE_V}x_{V'}$, 
\item
$H_V: = \left(  {1-\frac{1}{q} \over 1-\widetilde{\mu}/\gamma}\right)  \sum_{V':V'\in\mH_V} (K-dim(V')+\mu/\gamma) x_{V'}$. 
\end{itemize}
The choice of $H_V$ here is motivated by Remark \ref{remark_Lyapunov_function}.
The proof that this Lyapunov function works for proving Theorem \ref{thm_netcode}(b)
parallels the proof of Theorem \ref{thm:main}(b). 

\begin{remark} When network coding is considered, it is typically assumed that
peers do not exchange descriptions of the pieces they already have.   This is likely because such
descriptions are more complex than simple bit vectors indicating data pieces used without network coding, and
because network coding works quite well even without such exchange.   If exchange of information were used,
then any time a peer $A$ with subspace $V_A$ transfers a piece to a peer $B$ with subspace $V_B$ such that
$V_A \not\subset V_B$,  a useful transfer could be achieved.  Theorem \ref{thm_netcode} remains true under this
mode of operation if $\widetilde{\mu}=\mu$ and $q\rightarrow\infty$ is taken in part (b), and the gap
between parts (a) and (b) shrinks to zero.
\end{remark}

\subsection{Modeling faster recovery for unsuccessful contacts}  

One aspect of the model is that the time between upload attempts by a peer or a seed
do not depend on whether the attempts are successful.   In practice, it can be
expected that if an attempt is not successful because there is no useful piece
to transfer, then the time to the next attempt can be reduced, perhaps by some
constant factor $\eta > 1,$   such as $\eta=10.$  We discuss briefly how this might be addressed in the
model of this paper, and the implications.   The model of this paper is push oriented, in
that the times that peers and the fixed seed attempt uploads are generated by
their internal Poisson clocks.   If we assume that each of those clocks runs faster by a
factor $\eta$ until the next clock tick, whenever there is no useful piece to upload,
then two things happen when there is a large one club.   First, the rate of download
opportunities for a young, gifted, or infected peer increases by a factor close to $\eta,$  which
is probably a violation of an implied soft download constraint in our model.  Secondly,
this would worsen the missing piece syndrome if some of the peers arrive with
pieces at time of arrival (i.e. if there are gifted peers) because those peers would be
uploading piece one a factor $\eta$ more slowly than they would be downloading other pieces.
Their contribution to the upload rate of piece one before they  become
peer seeds would thus be reduced by a factor $\eta.$

For the original model, it would be mathematically equivalent for peer-to-peer contacts
to be modeled as pulls, with a peer randomly contacting another peer to download from
at the times of an internal rate $\mu$ Poisson clock  (while the fixed seed would still push
pieces).    If each peer attempting to download a piece
would run its clock faster by a factor $\eta > 1$ following an unsuccessful attempt, until the
next attempt, then again two things happen when there is a large one club.   First, the rate of upload
opportunities for a young, gifted, or infected peer increases by a factor close to $\eta,$  which
is probably a violation of the implied soft upload constraint in our model.  Secondly,
this would lessen the missing piece syndrome if some of the peers arrive with
pieces at time of arrival (i.e. if there are gifted peers) because those peers would be
uploading piece one a factor $\eta$ more quickly than they would be downloading other pieces.
Their contribution to the upload rate of piece one before they  become
peer seeds would thus be increased by a factor $\eta.$

Either of the above two approaches leads to a violation of our implicit assumption that peers
upload and download at the same rates.  Further, if no peers arrive with pieces (so there are
no gifted peers) the stability condition wouldn't change anyway.   A third approach would be to
consider a push or pull scenario, but explicitly limit upload and download rates at a peer to
be equal when they occur simultaneously.    Specifically, if some peers arrive with pieces, so there exist
gifted peers, then in the original model those gifted peers tend to have successful uploads--they
have the rare piece one to give to other peers--and successful downloads--they need pieces other
than piece one that most of the other peers have.   Having those gifted peers upload and download
at equal rates would preserve the balance implicit in our original model, and the condition for
stability would remain unchanged.

\subsection{The Borderline of Stability}    \label{sec:borderline}

Theorem \ref{thm:main} provides a sufficient condition for stability
and a matching sufficient condition for instability, but it leaves open the
borderline case: namely, when equality holds in  \eqref{eq:pcd} (or, equivalently,
\eqref{eq:qcd}) for one or more values of $k\in {\cal F}$ and the
strict inequality \eqref{eq:pcd} holds for all  other $k.$
While it may not be interesting from a practical point of view,  we comment
on the borderline case.   First, we give a precise result
for a limiting case of the original system, and then we offer
a conjecture.    As in  \cite{HajekZhu10_full}, a simpler network model results by taking a limit
as $\mu \rightarrow \infty.$  
Call a state {\em slow} if all peers in the system have the same type, which includes
the state such that there are no peers in the system.  Otherwise, call a state {\em fast}. 
The total rate of transition out of any slow state does not depend on $\mu,$ and the
total rate out of any fast state is bounded below by a positive constant times $\mu.$
For very large values of $\mu,$  the process spends most of its time in slow states.
The original Markov process can be transformed into a new one by {\em watching} the
original process while it is in the set of slow states.   This means removing
the portions of each sample path during which the process is in fast states,
and time-shifting the remaining parts of the sample path to leave no
gaps in time.  The limiting Markov process, which we call the
$\mu=\infty$ process, is the weak limit (defined as usual for probability measures
on the space of c\`{a}dl\`{a}g sample paths equipped
with the Skorohod topology)  of the original process
watched in the set of slow states, as $\mu \rightarrow \infty.$
If $\gamma$ is fixed as $\mu\rightarrow\infty$ the model becomes
degenerate, because a single peer seed would quickly convert all
other peers into peer seeds.   If $\gamma=\theta\mu$ for fixed $\theta$
and $\mu\rightarrow\infty$ then the $\mu=\infty$ model is more interesting
but somewhat complicated.   So we consider $\gamma=\infty$ for simplicity.
For further simplicity we consider networks of the form in Example 3 (for any $K\geq 2$)
for symmetric arrival rates.  Thus $\lambda_C=\lambda$ if $|C|=1$ and $\lambda_C=0$ otherwise.
Also,  $U_s=0$ (no fixed seed) and $\gamma=\infty.$  Note that these networks
are borderline cases, not covered by Theorem \ref{thm:main}.

By symmetry of the model, the state space of the $\mu=\infty$ process
can be reduced to
$\widehat{\cal S}=\{(0,0)\} \cup \{(n,k): n\geq 1, 1\leq k \leq K-1\},$
where a state $(n,k)$ corresponds to $n$ peers in the system which
all possess the same set of $k$ pieces.  State (0,0) is transient.
The transition rate diagram is pictured in Figure \ref{fig.mu_infinite}
for $K=3.$   
\begin{figure}
\post{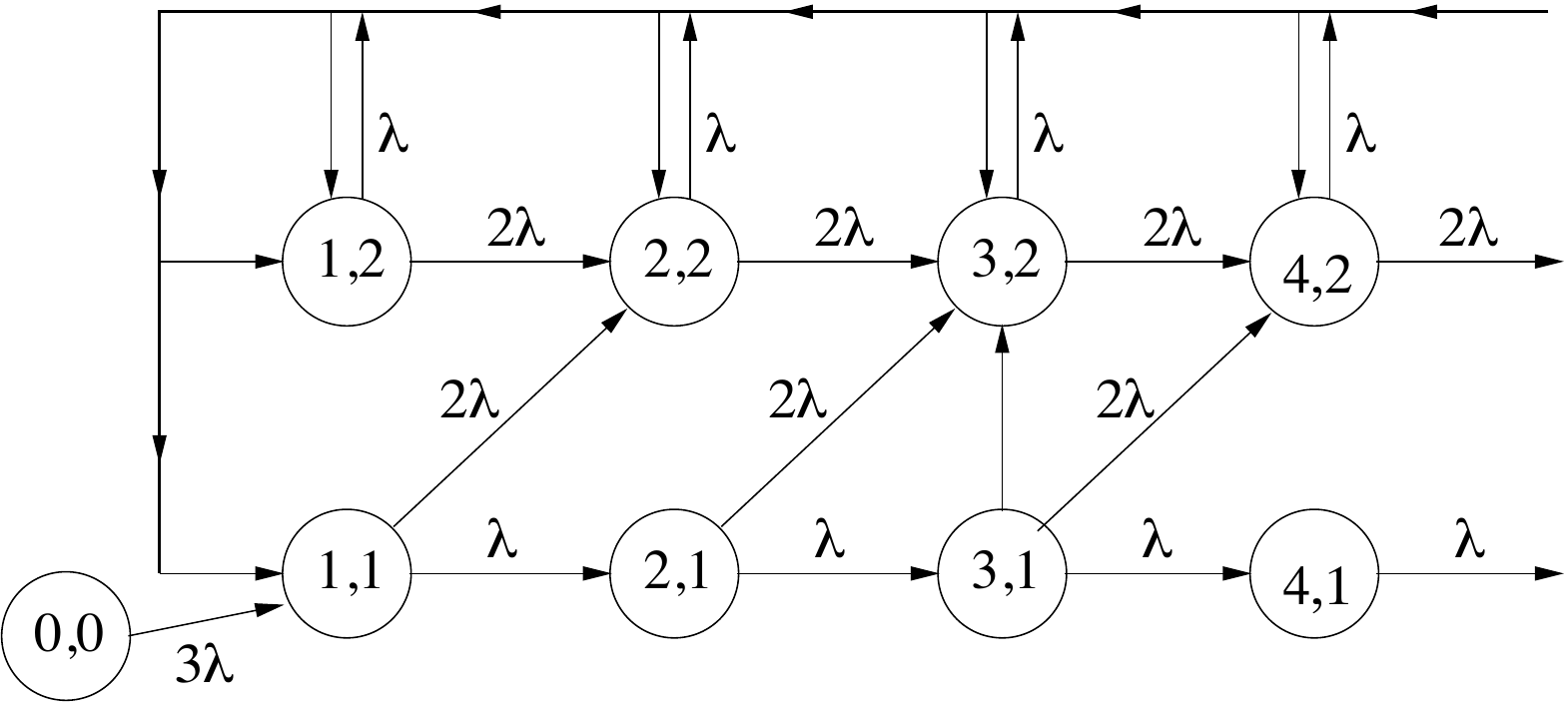}{8}
\caption{Transition rates of the $\mu=\infty$ variation of Example 3 with $\lambda_i=\lambda$ for all $i.$}
\label{fig.mu_infinite}
\end{figure}
States of the form $(n,K-1)$ form the top layer of states, and
are those for which all peers have the same set of $K-1$ pieces.
The transitions out of such a state $(n,K-1)$ is described as
follows.   There is a transition to state $(n+1,K-1)$ with rate $(K-1)\lambda,$
corresponding to the arrival of a new peer possessing one of the
 $K-1$ pieces that the other peers already have; the new peer instantly
 obtains all of the other $K-1$ pieces from the other peers.
At rate $\lambda$ a new peer arrives with the piece missing by all the
other peers.   The new peer downloads and uploads at equal rates, until
it either obtains all the $K-1$ other pieces, or until all the other peers
have departed.   By the nature of Poisson processes, the probability distribution
of the next state can be described in terms of flips of a fair coin, with ``heads" denoting
an upload by the new peer and ``tails" denoting a download by the new peer.
Let $Z$ denote the number of ``heads" in an experiment
of repeated coin flips, when a fair coin is flipped until ``tails" is observed $K-1$ times.
Then $Z$ represents the potential number of peers already in the system
that can leave due to uploads from the new peer.   If $Z\leq n-1$, then
the next state is $(n-Z,K-1).$   If $Z\geq n$ then the new state will have
the form $(1,j)$ with $1\leq j \leq K-1,$  corresponding to the case that
all peers that were originally in the system depart, and the new peer remains.
(The distribution over $j$ can be computed easily but is not important.)  Note
that $E[Z]=K-1,$   so the rate $(K-1)\lambda$ of upward unit jumps is equal to
the mean rate $\lambda E[Z]$ of decrease due to downward jumps (ignoring the
lower boundary).   Thus, when the process is in the top layer of states, it
evolves as a stationary, independent increment process with zero drift. 
Such processes are null-recurrent, and therefore, {\em the $\mu=\infty$
process is null-recurrent.}

In essence, the $\mu=\infty$ process is simple because peers remain young for only
an instant; there are no infections of young peers by gifted peers.  If $\mu$
is finite, such infections effectively increase the departure rate, by roughly a constant
divided by the number of peers in the system.  The constant is decreasing in $\mu.$
A reflecting Brownian motion with negative drift inversely
proportional to the state is positive recurrent if the constant of proportionality is
sufficiently large, and is null recurrent otherwise.  So the use of a diffusion approximation
leads us to pose the following conjecture, which pertains to the symmetric flat-network model
considered in \cite{MassoulieVojnovic08}:
\begin{conjecture}
Let $K\geq 1$ and suppose $\lambda_C=\lambda$ for $|C|=1$ and $\lambda_C=0$ otherwise.
For some $a_K > 0,$   the process is positive recurrent if $0 < \mu/\lambda < a_K$
and is null recurrent if $\mu/\lambda > a_K.$
\end{conjecture}

\section{Conclusion \label{sec:conclusion}}
By focusing on the missing piece syndrome, which affects the performance of a P2P system,
we have identified the minimum seed dwell times needed to stabilize the system. The model
includes a fixed seed, peers arriving with pieces, and seeds dwelling for a while as peer seeds
after obtaining the complete file.   It is a mathematical simplification of a P2P system
during the period of several hours or days after a flash crowd initiation of a file transfer, when
the arrival of new peers is relatively steady.   Our result identifies the stability region under all
possible rates of arrival, mean  times between transfer attempts, and distribution of pieces
brought in by new peers.   For tractability, we assumed that the times between upload attempts and the dwell times
of seeds are exponentially distributed random variables.   However,  we conjecture the
results hold for more general distributions; the instability half of our proof does
not rely on the assumption of exponential distributions.    Theorem 1 and its extensions
given here hold under the stated specific modeling assumptions.   It is our hope that the
results help with intuition and can be adapted to many other scenarios, including such effects as
heterogeneous link speeds or network topologies other than the fully connected one.

We summarize what can be taken away from our analysis, and point to
future work.   The first point is that stability can be
achieved (within the confines of the model) if peers remain in the system
a relatively short amount of time--no longer than the time needed to
upload one piece after obtaining a complete collection.  A second point is
that network coding can significantly lessen the effect of the missing piece
syndrome in the case that some peers are given pieces (random linear
combinations of data pieces) upon arrival.

A third point is that the stability condition is insensitive to the piece selection
policy,  and to network coding if peers don't arrive with pieces (i.e. no gifted peers).
However,  some systems that are provably unstable in the sense that they are
modeled by transient Markov processes, can be well behaved over long periods
of time in practice.  There may be a quasi-stable portion of the state space in which
the process dwells for a long time before the onset of a large one club occurs.   The use of network
coding or choice of piece selection policies can have a large impact on how long it
takes the system to enter a state with a large one club--a possible study for future
work would be to explore the longevity of a quasi-equilibrium in good network states.

\section{Appendix}

Miscellaneous results and a definition used in the main part of the paper are collected in this appendix.

\begin{prop}   \label{cor.FosterCompContinuous}
{\em Combined Foster-Lyapunov stability criterion and moment bound--continuous time} (See
\cite{Hajek567,MeynTweedie09}.)
Suppose $X$ is a continuous-time, irreducible Markov process on a countable state space ${\cal S}$ with generator matrix $Q.$
Suppose $V$, $f$, and $g$ are nonnegative functions on $\cal S$ such that
$QV(\mathbf{x})   \leq -f(\mathbf{x}) +g(\mathbf{x})$ for all $\mathbf{x}\in {\cal S}$, and, for some $\delta > 0$,  the set $C$ defined by
$C=\{ \mathbf{x} : f(\mathbf{x}) < g(\mathbf{x})+\delta\}$ is finite.  Suppose also that $\{ \mathbf{x} : V(\mathbf{x}) \leq K\}$ is finite for all $K$.
Then $X$ is positive recurrent and, if $\pi$ denotes the equilibrium distribution,   $\sum_\mathbf{x}  f(\mathbf{x})\pi(\mathbf{x})
 \leq \sum_\mathbf{x}  g(\mathbf{x})\pi(\mathbf{x})$.
\end{prop}

\begin{lemma} \label{lem:expansion}
{\em Bounding the drift of a function of a function of the state}
Suppose $X$ is a continuous-time, irreducible Markov process with countable state space $\mathcal{S}$ and with generator matrix $Q=\left(q(x,x'), x,x'\in\mathcal{S}\right)$. Suppose $f:\mathcal{S}\rightarrow  [0,\infty)$ and $V:\mathbb{R}\rightarrow [0,\infty)$ are two nonnegative functions; and suppose $V$ is differentiable with derivative $V'$ that is Lipschitz continuous
with Lipschitz constant $M$.    Then  $QV(f)$, the drift of $V(f)$, satisfies
\ben
\lefteqn{QV(f)(z) }\\
&:=& \sum_{z':z'\in\mathcal{S},z'\neq z} q(z,z')\left[V(f(z')) - V(f(z))\right] \\
&\leq& V'(f(z))Qf(z) +\\&& {M\over2}\sum_{z':z'\in\mathcal{S},z'\neq z}q(z,z')[f(z')-f(z)]^2,
\een
for all $z\in\mathcal{S},$ where $Qf$ is the drift of $f$.
\end{lemma}
\begin{proof}
The lemma follows from:
\ben
\lefteqn{V(f(z')) - V(f(z))}\\ &=& \int_{f(z)}^{f(z')} V'(x) dx \\
& = & \int_{f(z)}^{f(z')}\left[V'(f(z))+ \left(V'(x)-V'(f(z))\right)\right] dx\\
&\leq& \int_{f(z)}^{f(z')}  V'(f(z) )  dx  +  \bigg|   \int_{f(z)}^{f(z')}  M|x-f(z)| dx   \bigg|  
 \\
&=& V'(f(z))[f(z')-f(z)]+ {M\over2}[f(z')-f(z)]^2.\een
\end{proof}

\begin{definition} {\em Stochastic domination or coupling}
Suppose $A=(A_t : t\geq 0)$ and  $B=(B_t : t\geq 0)$  are two random processes, either both
discrete-time random processes, or both continuous time  random processes having right-continuous with
left limits sample paths.   Then $A$ is {\em stochastically dominated} by $B$ if there is a single probability
space $(\Omega, {\cal F}, P)$, and two random processes $\tilde{A}$ and $\tilde{B}$ on
 $(\Omega, {\cal F}, P)$, such that\\
\parbox{3.5in}{
(a) $A,\tilde{A}$ have the same finite dimensional  distributions,   \\
(b) $B, \tilde{B}$ have the same finite dimensional  distributions,   \\
and (c)  $P\{\tilde{A}_t \leq \tilde{B}_t~\mbox{for all}~t\}=1.$
}
 \end{definition}
 Clearly if $A$ is stochastically dominated by $B$, then for any $a$ and $t$, 
  $P\{ A_t \geq a\}  \leq P\{ B_t \geq a\}.$

\begin{proposition}  \label{prop:compoundKingman} 
{\em Kingman's Moment bound adapted to compound Poisson processes}
 (\cite{Kingman62}, see \cite{HajekZhu10_full})
Let $C$ be a compound Poisson process with $C_0=0$, with jump times given by a Poisson process
of rate $\alpha$, and jump sizes having mean $m_1$ and mean square value $m_2.$
Then for all $B >0$ and $\epsilon > \alpha m_1$
\begin{equation}   \label{eq:compoundKingman}
P\{ C_t   < B+\epsilon t ~ \mbox{for all}~ t \} \geq 1 -   \frac{\alpha m_2}{2B(\epsilon- \alpha m_1)} .  
\end{equation}
\end{proposition}

 \begin{lemma}   \label{lemma.mginfty} {\em A maximal bound for an $\mathbf M/GI/\infty$ queue}  (\cite{HajekZhu10_full})
 Let $M$ denote the number of customers in an $M/GI/\infty$ queueing system, with arrival rate
 $\lambda$ and mean service time $m.$    Suppose that $M_0=0.$  Then for $B,\epsilon > 0$,
 \begin{equation}  \label{eq.mginfty}
 P\{ M_t \geq B+\epsilon t ~~\mbox{for some}~t\geq 0\}  \leq   \frac{ e^{\lambda(m+1) } 2^{-B} }{1-2^{-\epsilon}}.
 \end{equation}
 \end{lemma}


\end{document}